\documentclass[article]{IEEEtran}
\usepackage{epsfig,latexsym,amsmath,amssymb,setspace,cite,color,graphicx,algorithm,algorithmic,cases}
\usepackage[caption=false,font=footnotesize]{subfig}
\usepackage[percent]{overpic}
\newcommand\numberthis{\addtocounter{equation}{1}\tag{\theequation}}

\usepackage{physics}
\usepackage{amsmath}
\usepackage{amsfonts}
\usepackage{amssymb}
\usepackage{dsfont}
\usepackage{bm}
\usepackage{amsthm}
\usepackage{newlfont}
\usepackage{float}
\usepackage{hyperref}
\usepackage{algorithm}
\usepackage{algorithmic}
\usepackage{enumerate}
\usepackage{chngcntr}
\usepackage{mathtools}
\usepackage{breqn}
\usepackage{stmaryrd}
\usepackage{cite}
\usepackage[yyyymmdd]{datetime}
\hypersetup{
    colorlinks=true, 
    linkcolor= blue, 
    citecolor=blue 
}

\begin{document}
\sloppy
\allowdisplaybreaks[1]

\newtheorem{thm}{Theorem} 
\newtheorem{lem}{Lemma}
\newtheorem{prop}{Proposition}
\newtheorem{cor}{Corollary}
\newtheorem{defn}{Definition}
\newcommand{\remarkend}{\IEEEQEDopen}
\newtheorem{remark}{Remark}
\newtheorem{rem}{Remark}
\newtheorem{ex}{Example}
\newtheorem{pro}{Property}

\newenvironment{example}[1][Example]{\begin{trivlist}
\item[\hskip \labelsep {\bfseries #1}]}{\end{trivlist}}
  
\renewcommand{\qedsymbol}{ \begin{tiny}$\blacksquare$ \end{tiny} }

\renewcommand{\leq}{\leqslant}
\renewcommand{\geq}{\geqslant}

\title {Function Computation Without Secure Links: Information and Leakage Rates}

\author{ R\'emi A. Chou, Joerg Kliewer \thanks{R. Chou is with the Department of Computer Science and Engineering, The University of Texas at Arlington, Arlington, TX. J. Kliewer is with the Department of Electrical and Computer
Engineering, New Jersey Institute of Technology, Newark, NJ.  This work was supported in part by NSF grants CCF-2201824 and CCF-2425371. E-mails: remi.chou@uta.edu, jkliewer@njit.edu. A preliminary version of this work has been presented at the 2022 IEEE International Symposium on
Information Theory (ISIT)~\cite{chou2022function}.}} 

\maketitle
\begin{abstract}
Consider $L$ users, who each hold private   data, and one fusion center who must compute a function of the private data of the $L$ users. To accomplish this task, each user may utilize a public and noiseless broadcast channel in a non-interactive manner. In this setting, and in the absence of any additional resources such as secure links, we study the optimal communication rates and minimum information leakages on the private user data that are achievable. Specifically, we study the information leakage of the user data at the fusion center (beyond the knowledge of the function output), as well as at predefined groups of colluding users who  eavesdrop one another. We derive the capacity region when the user data is independent, and inner and outer regions for the capacity region when the  user data is correlated.
\end{abstract} 

 \section{Introduction}

In this paper, we consider a function computation setting where the users do not have access to secure links to communicate among themselves but only to a public and noiseless broadcast channel. This setting contrasts with traditional information-theoretically secure multiparty computation settings, e.g.,~\cite{cramer2015secure} and references therein, where each pair of users has unlimited access to an information-theoretically secure communication link. Without this assumption, and in the absence of additional resources, perfect information-theoretic security of user data is impossible to obtain for the computation of arbitrary functions. In this context, our goal is to (i)~determine the level of privacy that is attainable, i.e., quantify the  minimum information leakage on the private user data that is achievable, and (ii) determine optimal communication rates for the computation of arbitrary functions.

In our setting, we consider $L$ users who each hold private data, and one fusion center who must compute a function of their data. In this function computation setting, the inputs of the function are sequences of
independent and identically distributed random variables. Each user can send one message, i.e., an encoded version of their data, over the public and noiseless broadcast channel. We are then interested in  characterizing the optimal communication rates, as well as the minimum information leakage on the private user data that is achievable. We distinguish two types of information leakage. The first one is the amount of information that the fusion center can learn from the public communication about the user data, beyond the knowledge of the function output. The second one is the amount of information that a group of colluding users can learn from the public communication about the data of all the other users. 
The private data of a given user is modeled by a sequence of independent and identically distributed random variables. We derive a capacity result when the data of the users is independent, and inner and outer regions for the capacity region when the data is correlated.

To sum up, our paper studies the trade-off between communication complexity and privacy leakage in a function computation setting. This departs from previous works that consider similar models, e.g.,~\cite{orlitsky2001coding,sefidgaran2016distributed}, but in the 
the absence of any privacy or security constraints.  We note that more advanced settings for function computation without security or privacy constraints have also considered interactive communication, e.g.,~\cite{ma2011some,ma2012interactive,ma2013infinite}.  Because no secure links are available in our model and a controlled privacy leakage is allowed, it also contrasts with other works, e.g., \cite{data2016communication,lee2014two}, that have considered settings related to secure multiparty computation   but with the additional assumptions that any pair of users can interactively communicate over secure noiseless links, and the additional requirement that no user data leakage is allowed.  Specifically,~\cite{lee2014two} studies the minimum amount of randomness needed at the users to perform secure addition, and \cite{data2016communication} studies optimal communication rates for function computation among three users -- full characterization of such optimal communication rates have been established for the computation of a few functions but  remains an open problem in general. While \cite{data2016communication,lee2014two} consider settings that do not permit data leakage, we note that \cite{chou2024private} and \cite{morteza2024distributed} allow controlled leakage, specifically for private addition in \cite{chou2024private} and private matrix multiplication in \cite{morteza2024distributed}. Another line of work has focused on function computation models when no user data leakage is allowed but in the absence of secure links, e.g., \cite{data2020interactive,tyagi2011function}. However, in such settings not all functions can be computed. 
Finally, the closest setting to our model in this paper is function computation with privacy constraints in \cite{tu2019function}, which also considers that the inputs of the function are sequences of independent and identically distributed random variables, and communication links among users are noiseless and unsecured.   The main differences between the model in \cite{tu2019function} and our model is that \cite{tu2019function} focuses on the computation of functions with three inputs and  considers a single external eavesdropper, whereas our model considers computation of functions with an arbitrary number of inputs and groups of colluding users who eavesdrop one another. To handle the achievability in  our multiuser and multi-eavesdropper setting, we heavily rely on the polymatroid structure of our achievability region and focus on a random binning approach. Our choice of a random binning approach contrasts with the approach in \cite{tu2019function}, which solely relies on typicality analysis and makes the analysis of leakage constraints more involved. In contrast, random binning relies on distribution approximation, which allows precise control over the distribution induced by the coding scheme and thereby facilitates the analysis of the leakage constraints, as it will be seen in the proofs. 

We also note that the line of research on secure rate-distortion, which examines trade-offs among the level of equivocation at an eavesdropper, the encoder output rate, and the distortion between the reconstructed and original source sequences, is closely related to our work. Foundational results in this area are presented in \cite{Yamamoto1983}. Scenarios involving a shared secret key are studied in \cite{Yamamoto1997,Merhav2008,Schieler2014}. Multiuser settings with multiple encoders are explored in \cite{Villard2013,Naghibi2015}, and the use of access structures is considered in \cite{zivarifard2024secure}. 
While our work shares the emphasis on trade-offs between privacy, encoder rate, and distortion, it differs in that the goal is to reconstruct a function of the source, rather than the source itself.

The remainder of the paper is organized as follows. We formalize the studied setting in Section~\ref{secps}. We present our main results in Section~\ref{secmr}. We propose in Section~\ref{secv} a variant of the model of Section~\ref{secps}. We present the proof of our main achievability result in Section~\ref{Secthach2} and relegate the other proofs to the appendix. Finally, we provide concluding remarks in Section~\ref{secconcl}.

\section{Problem Statement} 
In this paper, we consider two related but distinct multiuser function computation models as described in Sections \ref{secps} and~\ref{secv}.

Notation:  For $a,b \in \mathbb{R}$, define $\llbracket a,b \rrbracket \triangleq [\lfloor a \rfloor , \lceil b \rceil ] \cap \mathbb{N}$ and $[a]^+\triangleq\max(0,a)$. Unless specified otherwise, random variables and their realizations are denoted by uppercase and corresponding lowercase letters, respectively, e.g., $x$ is a realization of the random variable $X$. For two distributions $p$ and $q$ defined over the set $\mathcal{X}$, the variational distance between $p$ and $q$ is $\mathbb{V}(p,q) \triangleq \sum_{x\in\mathcal{X}}\vert p(x)-q(x) \vert $. For a set $\mathcal{S}$, $2^{\mathcal{S}}$ denotes the power set of $\mathcal{S}$.

\subsection{Model 1: Fusion center model} \label{secps}
Consider $L$ users indexed in $\mathcal{L} \triangleq \llbracket 1, L \rrbracket$. Consider $L$ finite alphabets $(\mathcal{X}_l)_{l \in \mathcal{L}}$ and a probability distribution $p_{X_{\mathcal{L}}}$ defined over $\mathcal{X}_{\mathcal{L}} \triangleq \bigtimes_{l\in \mathcal{L}} \mathcal{X}_l$ with the notation $X_{\mathcal{L}} \triangleq (X_l)_{l\in\mathcal{L}}$.  Consider $X^n_{\mathcal{L}}$ distributed according to $\prod_{i=1}^n p_{X_{\mathcal{L}}}$ and the notation $X^n_{\mathcal{L}} \triangleq (X_{\mathcal{L},i})_{i\in \llbracket 1, n \rrbracket}$. Assume that $X_l^n$ corresponds to an input available at User~$l\in\mathcal{L}$. For a function $f : \mathcal{X}_{\mathcal{L}} \to \mathcal{F}$, define $F \triangleq f( X_{\mathcal{L}} )$, $F^n \triangleq (f( X_{\mathcal{L},i} ))_{i \in \llbracket 1, n \rrbracket}$, and assume that this function needs to be computed at a fusion center. Assume that there is a public and noiseless broadcast channel from the users to the fusion center. Finally, let $\mathbb{A}$ be an arbitrary and fixed  set of non-empty, possibly overlapping, subsets of users, such that any set of colluding users $\mathcal{A} \in \mathbb{A}$ is  interested in learning the inputs of the other users in $\mathcal{A}^c$ from the public communication. For instance, if $L=3$ and $\mathbb{A} \triangleq \{ \{1,2\}, \{2,3\} \}$, then Users $1$ and $2$ (resp. $2$ and $3$) could be interested in colluding to learn information about the private data of User $3$ (resp. $1$). The setting is depicted in Figure~\ref{fig0}.
\begin{figure}[h]
\centering
  \includegraphics[width=7.5 cm]{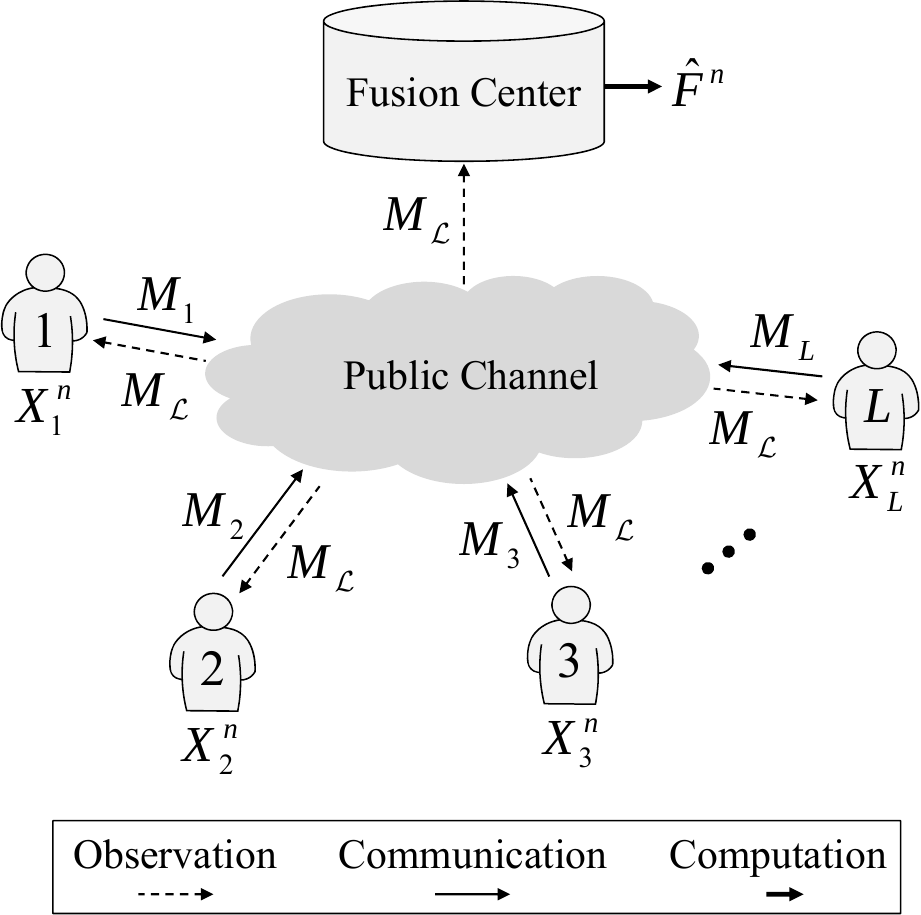}
  \caption{Function computation over a public and noiseless broadcast channel. $M_{\mathcal{L}} \triangleq (M_l)_{l\in\mathcal{L}}$ is the overall public communication and $\hat{F}^n$ is an estimate of $F^n \triangleq (f( X_{\mathcal{L},i} ))_{i \in \llbracket 1, n \rrbracket}$.} 
  \label{fig0}
\end{figure}
\begin{defn} \label{def1}
A $((2^{nR_l})_{l\in\mathcal{L}},n)$ computation scheme consists of
\begin{itemize}
\item $L$ messages sets $\mathcal{M}_l\triangleq \llbracket 1 , 2^{nR_l} \rrbracket$, $l\in \mathcal{L}$;
\item $L$ encoding functions $e_l : \mathcal{X}_l^n \to \mathcal{M}_l$, $l\in \mathcal{L}$;
\item One function $d : \mathcal{M}_{\mathcal{L}} \to \mathcal{F}^n$;
\end{itemize}
and operates as follows:
\begin{itemize}
\item User $l\in \mathcal{L}$ forms $M_l\triangleq e_l (X^n_l)$ and sends it to the fusion center over the public channel;
\item The fusion center forms an estimate $\widehat{F}^n \triangleq d(M_{\mathcal{L}})$ of $F^n$.
\end{itemize}
\end{defn}
\begin{defn}
A tuple $(({R_l})_{l\in\mathcal{L}},\Delta, (\Delta_{\mathcal{A}})_{\mathcal{A} \in \mathbb{A}})$ is achievable if there exists a sequence of $((2^{nR_l})_{l\in\mathcal{L}},n)$ computation schemes such that for any $\mathcal{A} \in \mathbb{A}$
\begin{align}
 \lim_{n\to \infty} \mathbb{P}[\widehat{F}^n \neq F^n] & = 0, \nonumber\\ &\text{ (function computation)} \label{eqrec} \\
 \lim_{n\to \infty} \frac{1}{n}I(X_{\mathcal{L}}^n ; M_{\mathcal{L}} | F^n) &\leq \Delta, \nonumber\\ & \text{ (input leakage at the fusion center)} \label{eqle} \\
 \lim_{n\to \infty}\frac{1}{n}I(X_{\mathcal{A}^c}^n ; M_{\mathcal{L}} |  X_{\mathcal{A}}^n) &\leq \Delta_{\mathcal{A}}, \nonumber\\ & \!\!\!\!\!\!\!\!\!\!\!\!\!\!\!\!\!\!\!\!\!\!\!\!\!\!\!\!\!\!\!\!\!\!\!\!\!\!\!\!\!\!\!\!\!\!\!\!\!\text{ (input leakage of users in $\mathcal{A}^c$ at colluding users in $\mathcal{A}$)} \label{eql2a}.
\end{align}
The set of all achievable tuples $(({R_l})_{l\in\mathcal{L}},\Delta, (\Delta_{\mathcal{A}})_{\mathcal{A} \in \mathbb{A}})$ is denoted by $\mathcal{C}_f(\mathbb{A})$. Note that the region $\mathcal{C}_f(\mathbb{A})$ is defined for a fixed function $f$. To simplify notation, we omit the subscript $f$ in the following.
\end{defn}

\eqref{eqrec} means that the fusion center obtains $F^n$ with a small probability of error.  \eqref{eqle} bounds the difference between $H(X_{\mathcal{L}}^n  |   F^n)$ and $H(X_{\mathcal{L}}^n | M_{\mathcal{L}}   F^n)$, i.e., quantifies the leakage of the inputs $X_{\mathcal{L}}^n$ at the fusion center  through the public communication $M_{\mathcal{L}}$, when accounting for the fact that the fusion center is supposed to learn~$F^n$.
Similarly,  for $\mathcal{A} \in \mathbb{A}$, \eqref{eql2a} bounds the difference between $H(X_{\mathcal{A}^c}^n  |   X_{\mathcal{A}}^n)$ and $H(X_{\mathcal{A}^c}^n | M_{\mathcal{L}}   X_{\mathcal{A}}^n)$, i.e., quantifies the leakage of the inputs $X_{\mathcal{A}^c}^n$ at the set of colluding users $\mathcal{A}$  through the public communication $M_{\mathcal{L}}$, when accounting for the fact that the set of colluding users $\mathcal{A}$ has access to $X_{\mathcal{A}}^n$. 

\begin{ex}
Suppose $\mathbb{A} = \{ \{l\} : l \in \mathcal{L} \}$. Then, the users do not collude but each user is curious about the inputs of all the other users.
\end{ex}

\begin{ex}
Suppose $\mathbb{A} = \emptyset$. Then, the users are not interested in learning the inputs of the other users.
\end{ex}

\begin{ex}
Let $\mathcal{E}$ and $\mathcal{O}$ be the sets of even and odd indices in $\mathcal{L}$, respectively. Suppose $\mathbb{A} = \{ \mathcal{E},\mathcal{O} \}$. Then, the users with odd (respectively even) indices are interested in colluding to learn the inputs of the users with even (respectively odd) indices.
\end{ex}

\subsection{Model 2: User-centric model}\label{secv}

\begin{figure}[h]
\centering
  \includegraphics[width=7.5 cm]{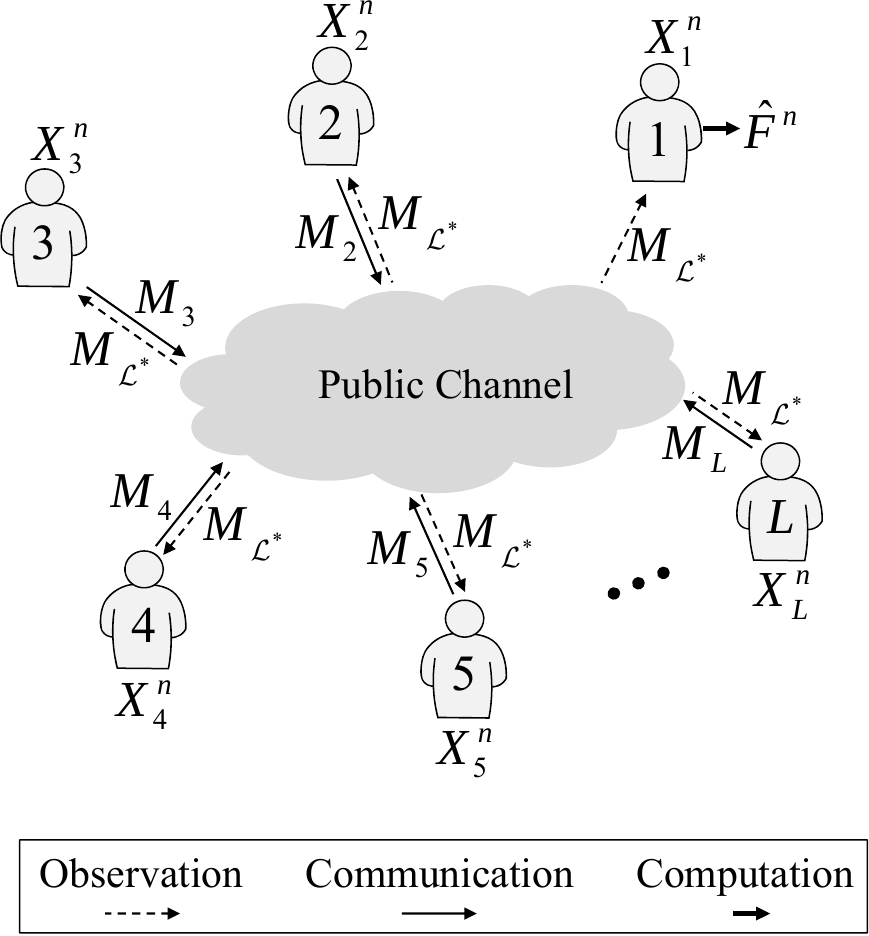}
  \caption{Function computation over a public and noiseless broadcast channel with $l_0 = 1$. $M_{\mathcal{L}^*} \triangleq (M_l)_{l\in\mathcal{L}^*}$ is the overall public communication, with $\mathcal{L^*} \triangleq \mathcal{L} \backslash \{ 1 \}$, and $\hat{F}^n$ is an estimate of $F^n \triangleq (f( X_{\mathcal{L},i} ))_{i \in \llbracket 1, n \rrbracket}$.}
  \label{fig2}
\end{figure}

Consider the same notation as in Section \ref{secps}. Fix $l_0 \in \mathcal{L}$ and define $ \mathcal{L}^* \triangleq \mathcal{L} \backslash \{ l_0 \}$.  For $\mathcal{A} \in \mathbb{A}$, define $\mathcal{A}^* \triangleq \mathcal{A} \backslash \{ l_0\}$ and $\mathcal{A}^{c*} \triangleq \mathcal{A}^c \backslash \{ l_0\}$. We consider the following variant of the setting of Section \ref{secps} as formalized in Definitions \ref{def1b}, \ref{def3}, and depicted in Figure \ref{fig2}. In this variant, there is no fusion center and a designated user (User $l_0$) needs to compute a function of all the users' private data including its own. Note that  the model of Section \ref{secps} is not a special case of the model described in this section because the condition \eqref{eqle} is not necessarily present in this section.

\begin{defn} \label{def1b}
A $((2^{nR_l})_{l\in\mathcal{L}^*},n)$ computation scheme consists of
\begin{itemize}
\item $L-1$ messages sets $\mathcal{M}_l\triangleq \llbracket 1 , 2^{nR_l} \rrbracket$, $l\in \mathcal{L}^*$;
\item $L-1$ encoding functions $e_l : \mathcal{X}_l^n \to \mathcal{M}_l$, $l\in \mathcal{L}^*$;
\item One function $d : \mathcal{M}_{\mathcal{L}^*} \times \mathcal{X}_{l_0}^n \to \mathcal{F}^n$;
\end{itemize}
and operates as follows:
\begin{itemize}
\item User $l\in \mathcal{L}^*$ forms $M_l \triangleq e_l(X_l^n)$ and sends it over the public channel;
\item User $l_0$ forms an estimate $\widehat{F}^n \triangleq d(M_{\mathcal{L}^*},X^n_{l_0})$ of $F^n$.
\end{itemize}
\end{defn}

\begin{defn} \label{def3}
A tuple $(({R_l})_{l\in\mathcal{L}^*}, (\Delta_{\mathcal{A}})_{\mathcal{A} \in \mathbb{A}})$ is achievable if there exists a sequence of $((2^{nR_l})_{l\in\mathcal{L}^*},n)$ computation schemes such that for any $\mathcal{A} \in \mathbb{A}$
\begin{align}
 \lim_{n\to \infty} \mathbb{P}[\widehat{F}^n  \neq F^n] & = 0,\nonumber\\ & \text{ (function computation)}  \label{eqdef21} \\
 \lim_{n\to \infty}\frac{1}{n}I(X_{\mathcal{A}^c}^n ; M_{\mathcal{L}^*} | \bar{F}^n X_{\mathcal{A}}^n) &\leq \Delta_{\mathcal{A}},\nonumber\\ & \!\!\!\!\!\!\!\!\!\!\!\!\!\!\!\!\!\!\!\!\!\!\!\!\!\!\!\!\!\!\!\!\!\!\!\!\!\!\!\!\!\!\!\!\!\!\!\!\!\!\!\!\!\!\!\!\!\!\!\!\!\!\!\text{ (input leakage of users in $\mathcal{A}^c$ at colluding users in $\mathcal{A}$)} \label{eql2b},
\end{align}
where $\bar{F}^n  \triangleq \begin{cases} F^n & \text{if } l_0 \in \mathcal{A}  \\ 
\emptyset & \text{if } l_0 \notin \mathcal{A} \end{cases}$. 
The set of all achievable tuples $(({R_l})_{l\in\mathcal{L}}, (\Delta_{\mathcal{A}})_{\mathcal{A} \in \mathbb{A}})$ is denoted by $\mathcal{C}_f(\mathbb{A},l_0)$. Note that the region $\mathcal{C}_f(\mathbb{A},l_0)$ is defined for a fixed function $f$. To simplify notation, we omit the subscript $f$ in the following.
\end{defn}

\eqref{eqdef21} means that User $l_0$ obtains $F^n$ with a small probability of error.  For $\mathcal{A} \in \mathbb{A}$, \eqref{eql2b} bounds the difference between $H(X_{\mathcal{A}^c}^n  |  \bar{F}^n X_{\mathcal{A}}^n)$ and $H(X_{\mathcal{A}^c}^n | \bar{F}^n M_{\mathcal{L}^*}   X_{\mathcal{A}}^n)$, i.e., quantifies the leakage of the inputs $X_{\mathcal{A}^c}^n$ at the set of colluding users $\mathcal{A}$  through the public communication $M_{\mathcal{L}^*}$, when accounting for the fact that the set of colluding users $\mathcal{A}$ has access to $(X_{\mathcal{A}}^n,\bar{F}^n)$.

\section{Main Results}  \label{secmr}

\subsection{Main results for Model 2}

We first present our outer and inner bounds on the capacity region  $\mathcal{C}(\mathbb{A}, l_0)$, $l_0 \in \mathcal{L}$, for arbitrarily correlated inputs.

\begin{thm}[Converse] \label{th5}
Fix $l_0 \in \mathcal{L}$ and define $\mathcal{O}(\mathbb{A},l_0)\triangleq \bigcap_{\mathcal{A} \in \mathbb{A}}\mathcal{R}(\mathcal{A},l_0)$, where $\mathcal{R}(\mathcal{A},l_0)$ is defined in \eqref{eqaddf}, and for $\mathcal{S} \subseteq \mathcal{L}$, we use the notation $R_{\mathcal{S}} \triangleq \sum_{l \in \mathcal{S}} R_l$, $X_{\mathcal{S}} \triangleq (X_l)_{l \in \mathcal{S}}$,  $U_{\mathcal{S}} \triangleq (U_l)_{l \in \mathcal{S}}$, and $V_{\mathcal{S}} \triangleq (V_l)_{l \in \mathcal{S}}$.
\begin{figure*}[t!]
\begin{align*}
\mathcal{R}(\mathcal{A},l_0) 
\triangleq &\{ ((R_l)_{l\in\mathcal{L}^*},  (\Delta_{\mathcal{A}'})_{\mathcal{A}'\in\mathbb{A}}) :  \\
 R_{\mathcal{S}}  
& \geq I( V_{\mathcal{S}} ; X_{\mathcal{S}} | V_{\mathcal{S}^c}X_{l_0}) - I(V_{\mathcal{S}};V_{\mathcal{S}^c}|X_{\mathcal{S}}X_{l_0})  + I(U_{\mathcal{S}};X_{\mathcal{S}}|U_{\mathcal{S}^c} V_{\mathcal{S}}X_{l_0}) - I(U_{\mathcal{S}};U_{\mathcal{S}^c}|X_{\mathcal{S}} V_{\mathcal{S}}X_{l_0}), \forall \mathcal{S} \subseteq \mathcal{L}^* ,\\
\Delta_{\mathcal{A}} &  \geq \begin{cases} I( U_{\mathcal{A}^{c}};X_{\mathcal{A}^{c}}|X_{\mathcal{A}})  -   I( X_{\mathcal{A}^c} ;F | X_{\mathcal{A}}) & \text{ if } l_0 \in \mathcal{A}\\
I( U_{\mathcal{A}^{c*}};X_{\mathcal{A}^{c*}}|X_{\mathcal{A}}) - [I( X_{l_0};U_{\mathcal{A}^{c*}}| V_{\mathcal{A}^{c*}} ) -I( X_{\mathcal{A}} ;U_{\mathcal{A}^{c*}}|V_{\mathcal{A}^{c*}} )]^+ & \text{ if } l_0 \notin \mathcal{A}  \end{cases} 
 ,\\ 
&\text{for some $(U_l,V_l)_{l\in\mathcal{L}^*}$ s.t. }  V_l - U_l - X_l - X_{\mathcal{L}\backslash\{l\}} , \forall l \in \mathcal{L}^*, \text{ and } H(F|X_{l_0}U_{\mathcal{L}^*})=0  
 \} . \numberthis \label{eqaddf}\end{align*}
 
\begin{align*}
\mathcal{R} (\mathbb{A}&,  p_{V_{\mathcal{L}^*} U_{\mathcal{L}^*}  X_{\mathcal{L}}}) \triangleq  \{ ((R_l)_{l\in\mathcal{L}^*},  (\Delta_{\mathcal{A}})_{\mathcal{A}\in\mathbb{A}}) : \\
 R_{\mathcal{S}}  
& \geq I( V_{\mathcal{S}} ; X_{\mathcal{S}} | V_{\mathcal{S}^c}X_{l_0}) - I(V_{\mathcal{S}};V_{\mathcal{S}^c}|X_{\mathcal{S}}X_{l_0})  + I(U_{\mathcal{S}};X_{\mathcal{S}}|U_{\mathcal{S}^c} V_{\mathcal{S}}X_{l_0}) - I(U_{\mathcal{S}};U_{\mathcal{S}^c}|X_{\mathcal{S}} V_{\mathcal{S}}X_{l_0}), \forall \mathcal{S} \subseteq \mathcal{L}^* ,\\
\Delta_{\mathcal{A}} &   \geq   I( U_{\mathcal{A}^{c}};X_{\mathcal{A}^{c}}|X_{\mathcal{A}})  -   I( X_{\mathcal{A}^c} ;F | X_{\mathcal{A}}), \forall \mathcal{A} \in \mathbb{A}, \text{ s.t. } \mathcal{A} \ni l_0 \\
\Delta_{\mathcal{A}} &  \geq I( U_{\mathcal{A}^{c*}};X_{\mathcal{A}^{c*}}|X_{\mathcal{A}}) - [I( X_{l_0};U_{\mathcal{A}^{c*}}| V_{\mathcal{A}^{c*}} ) -I( X_{\mathcal{A}} ;U_{\mathcal{A}^{c*}}|V_{\mathcal{A}^{c*}} )]^+,  \forall \mathcal{A} \in \mathbb{A}, \text{ s.t. }  \mathcal{A} \not\ni l_0      \}.\numberthis \label{eqaddf2}
\end{align*}
 
 \hrulefill
\end{figure*}

 Then, the following converse result holds $$\mathcal{C}(\mathbb{A},l_0) \subseteq \mathcal{O}(\mathbb{A},l_0).$$
\end{thm}
\begin{proof}
See Appendix \ref{Secth5}.
\end{proof}

Compared to Model 1, the study of the converse for Model~2 is more involved due to the presence of correlated information at the computation node. In particular, in the proof of Theorem~\ref{th5}, the auxiliary random variables depend  on $\mathcal{A} \in \mathbb{A}$, which requires us to consider the intersection $\bigcap_{\mathcal{A} \in \mathbb{A}}\mathcal{R}(\mathcal{A},l_0)$ to obtain an outer region on the capacity region $\mathcal{C}(\mathbb{A},l_0)$.

\begin{thm}[Achievability] 
\label{thach2}
Fix $l_0 \in \mathcal{L}$. Let $\mathcal{P}^{\mathcal{I}(l_0)}$ be the set of probability distributions $p_{V_{\mathcal{L}^*} U_{\mathcal{L}^*}  X_{\mathcal{L}}}$ over $\mathcal{V}_{\mathcal{L}^*} \times\mathcal{U}_{\mathcal{L}^*} \times \mathcal{X}_{\mathcal{L}}$ such that $p_{V_{\mathcal{L}^*} U_{\mathcal{L}^*}  X_{\mathcal{L}}} = p_{X_{\mathcal{L}}} \prod_{l\in\mathcal{L}^*}  p_{U_l|X_l} p_{V_l|U_l}$ and $H(F|U_{\mathcal{L}^*}X_{l_0})=0$. Then, define
$$\mathcal{I}(\mathbb{A},l_0) \triangleq \displaystyle\bigcup_{p_{V_{\mathcal{L}^*} U_{\mathcal{L}^*}  X_{\mathcal{L}}} \in \mathcal{P}^{\mathcal{I}(l_0)}} \mathcal{R} (\mathbb{A}, p_{V_{\mathcal{L}^*} U_{\mathcal{L}^*}  X_{\mathcal{L}}}),$$
where $\mathcal{R} (\mathbb{A},  p_{V_{\mathcal{L}^*} U_{\mathcal{L}^*}  X_{\mathcal{L}}})$ is defined in \eqref{eqaddf2}.

 Then, the following achievability result holds  $$\mathcal{C}(\mathbb{A},l_0) \supseteq \mathcal{I} (\mathbb{A},l_0).$$
\end{thm}
\begin{proof}
See Section \ref{Secthach2}.
\end{proof}

Unlike in the converse in Theorem \ref{th5}, the auxiliary random variables in the achievability region of Theorem \ref{thach2} do not depend on $\mathcal{A}\in\mathbb{A}$. Additionally, note that the condition $p_{V_{\mathcal{L}^*} U_{\mathcal{L}^*}  X_{\mathcal{L}}} = p_{X_{\mathcal{L}}} \prod_{l\in\mathcal{L}^*}  p_{U_l|X_l} p_{V_l|U_l}$ in Theorem \ref{thach2} is more restrictive than the condition $ V_l - U_l - X_l - X_{\mathcal{L}\backslash\{l\}} , \forall l \in \mathcal{L}^*$ in Theorem \ref{th5}. There is thus a gap between the achievability in Theorem \ref{thach2} and the converse in Theorem \ref{th5}. However, as shown in Theorem \ref{th4}, in the case of independent inputs at the users, we tighten the converse to obtain the capacity region and show the optimality of the coding strategy in the proof of Theorem~\ref{thach2}.

\begin{thm}[Capacity region for independent inputs]  \label{th4}
Fix $l_0 \in \mathcal{L}$. Assume that $p_{X_{\mathcal{L}}} = \prod_{l\in \mathcal{L}} p_{X_l}$. Let $\mathcal{P}$ be the set of probability distributions $p_{U_{\mathcal{L}^*}X_{\mathcal{L}}Q} =p_Q p_{X_{\mathcal{L}}} \prod_{l\in\mathcal{L}^*}  p_{U_l|X_lQ} $ over $\mathcal{U}_{\mathcal{L}^*} \times \mathcal{X}_{\mathcal{L}} \times \mathcal{Q}$ such that $H(F | U_{\mathcal{L}^*} X_{l_0} Q)=0$, $|\mathcal{U}_l| \leq |\mathcal{X}_l|, \forall l \in\mathcal{L}^*$, and $|\mathcal{Q}| \leq L+|\mathbb{A}| $. Then, the capacity region  is given by  
$$\mathcal{C}(\mathbb{A},l_0) = \displaystyle\bigcup_{p_{U_{\mathcal{L}^*}X_{\mathcal{L}}Q} \in \mathcal{P}} \mathcal{T} (\mathbb{A}, p_{U_{\mathcal{L}^*}X_{\mathcal{L}}Q}),$$
where 
\begin{align*}
& \!\!\!\!\!\!\!\!\!\!\!\!\!\!\!\!\!\!\!\!\!\!\!\!\!\!\!\!\mathcal{T} (\mathbb{A}, p_{U_{\mathcal{L}^*}X_{\mathcal{L}}Q}) \\ \triangleq \{ ((R_l&)_{l\in\mathcal{L}^*},  (\Delta_{\mathcal{A}})_{\mathcal{A} \in \mathbb{A}})
 : \forall l \in \mathcal{L}^*, \\
 R_{l}  
& \geq I( U_{l} ; X_{l}| Q) , \\
\Delta_{\mathcal{A}} & \geq I(X_{\mathcal{A}^c} ; U_{\mathcal{A}^c}| F Q X_{\mathcal{A}}) , \forall \mathcal{A} \in \mathbb{A}, \text{s.t. } \mathcal{A} \ni l_0 \\
\Delta_{\mathcal{A}} & \geq I(X_{\mathcal{A}^{c*}} ; U_{\mathcal{A}^{c*}}  |Q ), \forall \mathcal{A} \in \mathbb{A}, \text{s.t. } \mathcal{A} \not\ni l_0 
 \}.
\end{align*}
\end{thm}
\begin{proof}
See Appendix \ref{Secth4}.
\end{proof}
Note that the achievability part of Theorem~\ref{th4} follows similarly to that of Theorem~\ref{thach2} for correlated inputs. However, the converse of Theorem~\ref{th5}, which addresses the correlated inputs case, is not readily applicable to derive the converse of Theorem~\ref{th4}. For this reason, we develop a new converse using different definitions for the auxiliary random variables.
\begin{ex}
Assume $L=2$, $l_0=2$, $p_{X_1X_2}=p_{X_1}p_{X_2}$, and $\mathbb{A} = \{ \{2\}\}$.  Let $\mathcal{P}$ be the set of probability distributions $p_{U_1X_1X_2} = p_{X_1X_2}  p_{U_1|X_1} $ over $\mathcal{U}_1 \times \mathcal{X}_1 \times \mathcal{X}_2$ such that $H(F|U_1 X_2)=0$, $|\mathcal{U}_1| \leq |\mathcal{X}_1| $. Then,  the capacity region  is given by  
\begin{align*}
\mathcal{C} (\mathbb{A},l_0) 
&= \textstyle\bigcup_{p_{U_1X_1X_2} \in \mathcal{P}}    \{ (R_1,\Delta_{\{2\}} )
 :    
  R_1  
  \geq I( U_1 ; X_1) ,  \\ &\phantom{-------} 
\Delta_{\{2\}}    \geq I( U_1 ; X_1 ) - I( F ; X_1 |  X_2)     
 \} \\
&  =\{ (R_1,\Delta_{\{2\}} )
 :    
  R_1  
  \geq I( U_1^* ; X_1) ,  \\ &\phantom{-------}
\Delta_{\{2\}}    \geq   I(U_1^*;X_1) - I( F ; X_1 |  X_2) 
 \},
 \end{align*}
where $I( U_1^* ; X_1  ) \triangleq \min_{p_{U_1X_1} \in \mathcal{P}} I( U_1 ; X_1)$. The achievability proof follows from  Theorem \ref{th4}, and the converse proof is obtained by defining $U_{l,i} \triangleq (M_l,X^{i-1}_l,X_{l_0,\llbracket  1,n\rrbracket\backslash \{i\}})$, $i \in \llbracket  1,n\rrbracket$, and $U_l \triangleq (U_{l,Q},Q)$, where $Q$ is uniformly distributed over $\llbracket 1 , n \rrbracket$,  in the converse proof of Theorem~\ref{th4}.\end{ex}
As we can see in this example, there is  a linear relationship between the optimal communication rate $R_1$ and leakage $\Delta_{\{2\}} $ due to the fact that $I( F ; X_1 |  X_2) $ is constant.

\subsection{Main results for Model 1}
We first present our outer and inner bounds on the capacity region  $\mathcal{C}(\mathbb{A})$ for arbitrarily correlated inputs.

\begin{thm}[Converse] \label{th1}
Let $\mathcal{P}^{\mathcal{O}}$ be the set of probability distributions $p_{U_{\mathcal{L}}X_{\mathcal{L}}}$ over $\mathcal{U}_{\mathcal{L}} \times \mathcal{X}_{\mathcal{L}}$ such that for all $\mathcal{S} \subseteq \mathcal{L}$, $U_{\mathcal{S}} - X_{\mathcal{S}} - X_{\mathcal{L}}$ forms a Markov chain, and $H(F|U_{\mathcal{L}})=0$. Next, define
$$\mathcal{O}(\mathbb{A}) \triangleq \displaystyle\bigcup_{p_{U_{\mathcal{L}}X_{\mathcal{L}}} \in \mathcal{P}^{\mathcal{O}}} \mathcal{R} (\mathbb{A}, p_{U_{\mathcal{L}}X_{\mathcal{L}}}),$$
with 
\begin{align*}
\mathcal{R} (\mathbb{A}, p_{U_{\mathcal{L}}X_{\mathcal{L}}})\triangleq \{ (&(R_l)_{l\in\mathcal{L}}, \Delta, (\Delta_{\mathcal{A}})_{\mathcal{A} \in \mathbb{A}})\!
 : \!\forall \mathcal{S} \subseteq \mathcal{L}, \forall \mathcal{A} \in \mathbb{A}, \\
  R_{\mathcal{S}}  
& \geq I( U_{\mathcal{S}} ; X_{\mathcal{S}}| U_{\mathcal{S}^c}) - I(U_{\mathcal{S}};U_{\mathcal{S}^c}|X_{\mathcal{S}}) , \\
\Delta  & \geq I( U_{\mathcal{L}} ; X_{\mathcal{L}} | F ) , \\
\Delta_{\mathcal{A}} & \geq I(X_{\mathcal{A}^c} ; U_{\mathcal{A}^c} | X_{\mathcal{A}}   )    
 \}, \numberthis \label{region}
\end{align*}
where for $\mathcal{S} \subseteq \mathcal{L}$, we use the notation $R_{\mathcal{S}} \triangleq \sum_{l \in \mathcal{S}} R_l$, $X_{\mathcal{S}} \triangleq (X_l)_{l \in \mathcal{S}}$, and $U_{\mathcal{S}} \triangleq (U_l)_{l \in \mathcal{S}}$.
Then, the following converse result holds $$\mathcal{C}(\mathbb{A}) \subseteq \mathcal{O} (\mathbb{A}).$$
\end{thm}
\begin{proof}
See Appendix \ref{secth1}.
\end{proof}

\begin{thm}[Achievability]  \label{th2}
Let $\mathcal{P}^{\mathcal{I}}$ be the set of probability distributions $p_{U_{\mathcal{L}}X_{\mathcal{L}}}$ over $\mathcal{U}_{\mathcal{L}} \times \mathcal{X}_{\mathcal{L}}$ such that $p_{U_{\mathcal{L}}   X_{\mathcal{L}}} = p_{X_{\mathcal{L}}} \prod_{l\in\mathcal{L}}  p_{U_l|X_l}$ and $H(F|U_{\mathcal{L}})=0$. Next, define
$$\mathcal{I}(\mathbb{A}) \triangleq \displaystyle\bigcup_{p_{U_{\mathcal{L}}X_{\mathcal{L}}} \in \mathcal{P}^{\mathcal{I}}} \mathcal{R} (\mathbb{A}, p_{U_{\mathcal{L}}X_{\mathcal{L}}}),$$
where $\mathcal{R} (\mathbb{A}, p_{U_{\mathcal{L}}X_{\mathcal{L}}})$ is defined in \eqref{region}. Then, the following achievability result holds $$\mathcal{C}(\mathbb{A}) \supseteq \mathcal{I} (\mathbb{A}).$$
\end{thm}
\begin{proof}
See Appendix \ref{secth2}.
\end{proof}
Note that the condition $p_{U_{\mathcal{L}}   X_{\mathcal{L}}} = p_{X_{\mathcal{L}}} \prod_{l\in\mathcal{L}}  p_{U_l|X_l}$ in Theorem \ref{th2} is more restrictive than the condition $U_{\mathcal{S}} - X_{\mathcal{S}} - X_{\mathcal{L}}, \forall \mathcal{S} \subseteq \mathcal{L}$ in Theorem \ref{th1}. There is thus a gap between the achievability in Theorem \ref{th1} and the converse in Theorem \ref{th2}. However, as shown in Theorem \ref{th3}, in the case of independent inputs at the users, we tighten the converse to obtain the capacity region and show the optimality of the coding strategy in the proof of Theorem~\ref{th2}.

\begin{thm}[Capacity region for independent inputs]  \label{th3}
Assume that $p_{X_{\mathcal{L}}} = \prod_{l\in \mathcal{L}} p_{X_l}$. Let $\mathcal{P}$ be the set of probability distributions $p_{U_{\mathcal{L}}X_{\mathcal{L}}Q} =p_Q p_{X_{\mathcal{L}}} \prod_{l\in\mathcal{L}}  p_{U_l|X_lQ} $ over $\mathcal{U}_{\mathcal{L}} \times \mathcal{X}_{\mathcal{L}} \times \mathcal{Q}$ such that $H(F|U_{\mathcal{L}}Q)=0$, $|\mathcal{U}_l| \leq |\mathcal{X}_l|, \forall l \in\mathcal{L}$, and $|\mathcal{Q}| \leq L+|\mathbb{A}| + 2$. Then, the capacity region  is given by  
$$\mathcal{C}(\mathbb{A}) = \displaystyle\bigcup_{p_{U_{\mathcal{L}}X_{\mathcal{L}}Q} \in \mathcal{P}} \mathcal{T} (\mathbb{A}, p_{U_{\mathcal{L}}X_{\mathcal{L}}Q}),$$
where 
\begin{align*}
\mathcal{T} (\mathbb{A}, p_{U_{\mathcal{L}}X_{\mathcal{L}}Q})\!\triangleq \! \{ ((R_l)_{l\in\mathcal{L}},& \Delta, (\Delta_{\mathcal{A}})_{\mathcal{A} \in \mathbb{A}})
 : \forall l \! \in \! \mathcal{L}, \forall \mathcal{A} \! \in \!\mathbb{A}, \\
  R_{l}  
& \geq I( U_{l} ; X_{l}| Q) , \\
\Delta  & \geq I( U_{\mathcal{L}} ; X_{\mathcal{L}} | FQ ) , \\
\Delta_{\mathcal{A}} & \geq I(X_{\mathcal{A}^c} ; U_{\mathcal{A}^c} | Q  )    
 \}.
\end{align*}
 \end{thm}
\begin{proof}
See Appendix \ref{Secth3}.
\end{proof}
Note that the achievability part of Theorem~\ref{th3} follows similarly to that of Theorem~\ref{th1} for correlated inputs. However, the converse of Theorem~\ref{th2}, which addresses the correlated inputs case, is not readily applicable to derive the converse of Theorem~\ref{th3}. For this reason, we develop a new converse using different definitions for the auxiliary random variables.

\begin{ex}\label{ex4}
Assume $L=1$ and $\mathbb{A} = \emptyset$.  Let $\mathcal{P}$ be the set of probability distributions $p_{U_1X_1} = p_{X_1}  p_{U_1|X_1} $ over $\mathcal{U}_1 \times \mathcal{X}_1$ such that $H(F|U_1)=0$, $|\mathcal{U}_1| \leq |\mathcal{X}_1| $. Then,  the capacity region  is given by  
\begin{align*}
\mathcal{C} (\mathbb{A})
& = \textstyle\bigcup_{p_{U_1X_1} \in \mathcal{P}}    \{ (R_1,\Delta )
 :    
  R_1  
  \geq I( U_1 ; X_1) ,  \\
  & \phantom{---------llllll}
\Delta    \geq I( U_1 ; X_1 ) - I(X_1; F )  
 \}\\
 & =  \{ (R_1,\Delta )
 :    
  R_1  
  \geq I( U_1^* ; X_1) ,  \\
  & \phantom{----lllll-}
\Delta    \geq I( U_1^* ; X_1  )  -I(  X_1 ; F ) 
 \},
\end{align*}
where $I( U_1^* ; X_1  ) \triangleq \min_{p_{U_1X_1} \in \mathcal{P}} I( U_1 ; X_1)$. The achievability follows from  Theorem \ref{th3}, and the converse is obtained by defining $U_1 \triangleq (U_{1,Q},Q)$  in the converse proof of Theorem~\ref{th3}.
\end{ex}
As we can see in this example, there is a linear relationship between communication rate $R_1$ and leakage $\Delta$ as $I(X_1 ; F )$ is a constant term.

\section{Proof of Theorem \ref{thach2}} \label{Secthach2}
Without loss of generality, we assume that $l_0 \triangleq L$ in this section to simplify notation. 
\subsection{Proof overview}
For $p_{V_{\mathcal{L}^*} U_{\mathcal{L}^*}  X_{\mathcal{L}}} = p_{X_{\mathcal{L}}} \prod_{l\in\mathcal{L}^*}  p_{U_l|X_l} p_{V_l|U_l}$, we first prove in Section \ref{secreduction2}, that the achievability of $\mathcal{R} (\mathbb{A},{p_{V_{\mathcal{L}^*} U_{\mathcal{L}^*} X_{\mathcal{L}}}})$   -- that is, the achievability of any point in this region--  can be reduced to the achievability of the rate-tuple 
\begin{align}
&\mathbf{R}^{\star}(\mathbb{A},{p_{V_{\mathcal{L}^*} U_{\mathcal{L}^*} X_{\mathcal{L}}}}) \nonumber \\ &\triangleq ((I(V_{l};X_{l} | V_{1:l-1} X_{l_0})\nonumber \\ &\phantom{--} + I(U_{l};X_{l} | U_{1:l-1}V_{l:L-1} X_{l_0})  )_{l \in \mathcal{L}^*}, (\Delta_{\mathcal{A}})_{\mathcal{A} \in \mathbb{A}}), \label{eqredux2}
\end{align}
 where, for $l \in \mathcal{L}^*$, we use the notation $U_{1:l-1}  \triangleq U_{\llbracket 1, l-1 \rrbracket }$ and $V_{l:L-1} \triangleq V_{\llbracket l, L-1 \rrbracket } $ for convenience. Then, we provide a coding scheme  that relies on random binning in Section \ref{seccod2} to achieve this rate-tuple. 
 
\subsection{Reduction of the achievability of $\mathcal{R} (\mathbb{A},{p_{V_{\mathcal{L}^*} U_{\mathcal{L}^*}  X_{\mathcal{L}}}})$  to the achievability of  $\mathbf{R}^{\star}(\mathbb{A},{p_{V_{\mathcal{L}^*} U_{\mathcal{L}^*} X_{\mathcal{L}}}})$} \label{secreduction2}

We first review definitions and   properties related to  polymatroids \cite{edmonds1970submodular}.

\begin{defn}[\!\!\cite{edmonds1970submodular}]  
Let $g: 2^{\mathcal{L}^*} \to \mathbb{R}$.
\begin{enumerate} [(i)]
\item $g$ is submodular if $\forall \mathcal{S},\mathcal{T} \subseteq \mathcal{L}^*, g(\mathcal{S}\cup \mathcal{T}) +  g(\mathcal{S}\cap \mathcal{T}) \leq g(\mathcal{S})+ g(\mathcal{T})$.
\item $g$ is supermodular if $-g$ is submodular.
\end{enumerate}
\end{defn}

\begin{defn}[\!\!\cite{edmonds1970submodular}] \label{defcp}
Let $g: 2^{\mathcal{L}^*} \to \mathbb{R}$. $
\mathcal{P} (g) \triangleq \left\{ (R_{l})_{l \in \mathcal{L}^*} \in \mathbb{R}_+^{L-1} :  R_{\mathcal{S}} \geq g(\mathcal{S}) , \forall \mathcal{S} \subseteq \mathcal{L}^* \right\}
$
associated with the function $g$, is a contrapolymatroid if 
\begin{enumerate} [(i)]
\item $g$  is normalized, i.e., $g(\emptyset) =0$,
\item $g$  is non-decreasing, i.e., $\forall \mathcal{S}_1,\mathcal{S}_2 \subseteq \mathcal{L}^*, \mathcal{S}_1\subseteq\mathcal{S}_2 \implies g(\mathcal{S}_1) \leq g(\mathcal{S}_2)$,
\item $g$  is supermodular.
\end{enumerate}
\end{defn}

\begin{lem}[\!\!{\cite{edmonds1970submodular}}] \label{lempm}
Suppose $\mathcal{P} (g)$
associated with the function $g: 2^{\mathcal{L}^*} \to \mathbb{R}$  is a contrapolymatroid.
\begin{enumerate}[(i)]
\item  Any point in $\mathcal{P} (g)$ is dominated  by a point in the dominant face $$
\mathcal{D} (g) \triangleq \left\{ (R_{l})_{l \in \mathcal{L}^*} \in \mathcal{P} (g):  R_{\mathcal{L}^*} = g(\mathcal{L}^*)   \right\}.$$
\item By denoting the symmetric group on $\mathcal{L}^*$ by $\text{Sym}(L-1)$, we have
\begin{align*}
&\mathcal{D} (g_{p_{V_{\mathcal{L}^*} U_{\mathcal{L}^*}  X_{\mathcal{L}}}}) 
\\&\phantom{--} =   \text{Conv} \left(
  \left\{ (C_{\pi(l)})_{l \in \mathcal{L}^*} :\pi \in \text{Sym}(L-1) \right\} \right), 
\end{align*}
where for any $\pi \in \text{Sym}(L-1)$,  $l \in \mathcal{L}^*$,
$$C_{\pi(l)} \triangleq  g\left( \pi (l:L-1) \right) - g \left( \pi(l+1:L-1) \right), $$
with the notation $\pi (i:j) \triangleq \{ \pi(k) : k \in \llbracket i,j \rrbracket ) \}$ for $i,j \in \mathcal{L}^*$.
\end{enumerate}
\end{lem}
We now show the reduction of the achievability of $\mathcal{R} (\mathbb{A},{p_{V_{\mathcal{L}^*} U_{\mathcal{L}^*}  X_{\mathcal{L}}}})$  to the achievability of  $\mathbf{R}^{\star}(\mathbb{A},{p_{V_{\mathcal{L}^*} U_{\mathcal{L}^*} X_{\mathcal{L}}}})$ through a series of lemmas.

\begin{lem} \label{lem1}
Fix $p_{V_{\mathcal{L}^*} U_{\mathcal{L}^*}  X_{\mathcal{L}}} = p_{X_{\mathcal{L}}} \prod_{l\in\mathcal{L}^*}  p_{U_l|X_l} p_{V_l|U_l}$. Then, the set function $g_{p_{V_{\mathcal{L}^*} U_{\mathcal{L}^*}  X_{\mathcal{L}}}}$ is normalized, non-decreasing, and supermodular, where
\begin{align*}
g_{p_{V_{\mathcal{L}^*} U_{\mathcal{L}^*}  X_{\mathcal{L}}}} : 2^{\mathcal{L}^*} &\to \mathbb{R}\\
 \mathcal{S} & \mapsto I(V_{\mathcal{S}};X_{\mathcal{L}} | V_{\mathcal{S}^c} X_{l_0}) \\ & \phantom{--} +  I(U_{\mathcal{S}};X_{\mathcal{L}} | U_{\mathcal{S}^c} V_{\mathcal{S}} X_{l_0}).
\end{align*}
\end{lem}
\begin{proof} See Appendix \ref{app_lem}.\end{proof}

\begin{lem} \label{lem2}
Fix $p_{V_{\mathcal{L}^*} U_{\mathcal{L}^*}  X_{\mathcal{L}}} = p_{X_{\mathcal{L}}} \prod_{l\in\mathcal{L}^*}  p_{U_l|X_l} p_{V_l|U_l}$. Then, 
\begin{enumerate}[(i)]
\item \begin{align*}
&\!\!\!\!\!\!\mathcal{P} (g_{p_{V_{\mathcal{L}^*} U_{\mathcal{L}^*}  X_{\mathcal{L}}}}) \\
&\!\!\!\!\!\!\triangleq \left\{ (R_{l})_{l \in \mathcal{L}^*} \in \mathbb{R}_+^{L-1} :  R_{\mathcal{S}} \geq g_{p_{V_{\mathcal{L}^*} U_{\mathcal{L}^*}  X_{\mathcal{L}}}}(\mathcal{S}) , \forall \mathcal{S} \subseteq \mathcal{L}^* \right\}
\end{align*}
associated with the function $g_{p_{V_{\mathcal{L}^*} U_{\mathcal{L}^*}  X_{\mathcal{L}}}}$ defined in Lemma \ref{lem1}, is a contrapolymatroid.
\item  Any point in $\mathcal{P} (g_{p_{V_{\mathcal{L}^*} U_{\mathcal{L}^*}  X_{\mathcal{L}}}})$ is dominated  by a point in the dominant face $\mathcal{D} (g_{p_{V_{\mathcal{L}^*} U_{\mathcal{L}^*}  X_{\mathcal{L}}}})$, where
\begin{align*}
&\!\!\!\mathcal{D} (g_{p_{V_{\mathcal{L}^*} U_{\mathcal{L}^*}  X_{\mathcal{L}}}}) \\
 & \!\! \triangleq \left\{ \!(R_{l})_{l \in \mathcal{L}^*} \!\in
 \!\mathcal{P} (g_{p_{V_{\mathcal{L}^*} U_{\mathcal{L}^*}  X_{\mathcal{L}}}}\!)\!:\!  R_{\mathcal{L}^*}\!\! = g_{p_{V_{\mathcal{L}^*} U_{\mathcal{L}^*}  X_{\mathcal{L}}}}\!(\mathcal{L}^*) \!  \right\}\!.
\end{align*}
\item By denoting the symmetric group on $\mathcal{L}^*$ by $\text{Sym}(L-1)$, the dominant face has the following characterization:
\begin{align}
&\mathcal{D} (g_{p_{V_{\mathcal{L}^*} U_{\mathcal{L}^*}  X_{\mathcal{L}}}}) \nonumber\\
 & \phantom{--}=   \text{Conv} \left(
  \left\{ (C_{\pi(l)})_{l \in \mathcal{L}^*} :\pi \in \text{Sym}(L-1) \right\} \right), \label{eqdom}
\end{align}
where for any $\pi \in \text{Sym}(L-1)$,  $l \in \mathcal{L}^*$,
\begin{align*}C_{\pi(l)} &\triangleq I(V_{\pi(l)};X_{\pi(l)} | V_{\pi(1:l-1)}X_{l_0})\\
&\phantom{--}+I(U_{\pi(l)};X_{\pi(l)} | U_{\pi(1:l-1)}V_{\pi(l:L-1)}X_{l_0}) , \end{align*}
with the notation $\pi (i:j) \triangleq \{ \pi(k) : k \in \llbracket i,j \rrbracket ) \}$ for $i,j \in \mathcal{L}^*$.
\end{enumerate}
\end{lem}

\begin{proof}
See Appendix \ref{app_lem}.
\end{proof}

\begin{lem} \label{lem3}
Fix $p_{V_{\mathcal{L}^*} U_{\mathcal{L}^*}  X_{\mathcal{L}}} = p_{X_{\mathcal{L}}} \prod_{l\in\mathcal{L}^*}  p_{U_l|X_l} p_{V_l|U_l}$ and define
\begin{align*}
&\mathcal{R} (p_{V_{\mathcal{L}^*} U_{\mathcal{L}^*}  X_{\mathcal{L}}} ) \\
& \triangleq \left\{ (R_{l})_{l \in \mathcal{L}^*} \in \mathbb{R}_+^{L-1} :  R_{\mathcal{S}} \geq \bar{g}_{p_{V_{\mathcal{L}^*} U_{\mathcal{L}^*}  X_{\mathcal{L}}} }(\mathcal{S}) , \forall \mathcal{S} \subseteq \mathcal{L}^* \right\},
\end{align*}
where 
\begin{align*}
\bar{g}_{p_{V_{\mathcal{L}^*} U_{\mathcal{L}^*}  X_{\mathcal{L}}} } : 2^{\mathcal{L}^*} &\to \mathbb{R}\\
 \mathcal{S} & \mapsto I( V_{\mathcal{S}} ; X_{\mathcal{S}} | V_{\mathcal{S}^c}X_{l_0}) \\
 & \phantom{-ll}- I(V_{\mathcal{S}};V_{\mathcal{S}^c}|X_{\mathcal{S}}X_{l_0})  \\
 & \phantom{-ll}+ I(U_{\mathcal{S}};X_{\mathcal{S}}|U_{\mathcal{S}^c} V_{\mathcal{S}}X_{l_0}) \\
 & \phantom{-ll}- I(U_{\mathcal{S}};U_{\mathcal{S}^c}|X_{\mathcal{S}} V_{\mathcal{S}}X_{l_0}) .
\end{align*}
Then,
$$
\mathcal{R} ({p_{V_{\mathcal{L}^*} U_{\mathcal{L}^*}  X_{\mathcal{L}}} }) = \mathcal{P} (g_{p_{V_{\mathcal{L}^*} U_{\mathcal{L}^*}  X_{\mathcal{L}}} }) 
$$
\end{lem}
\begin{proof}  See Appendix \ref{app_lem}.\end{proof}

\begin{lem}
Suppose that there exists a sequence of computation schemes that achieve  the rate-tuple $\mathbf{R}^{\star}(\mathbb{A},{p_{V_{\mathcal{L}^*} U_{\mathcal{L}^*} X_{\mathcal{L}}}})$ defined in \eqref{eqredux2}. Then, there exists a sequence of computation schemes that achieve any point in   $\mathcal{R} (\mathbb{A},{p_{V_{\mathcal{L}^*} U_{\mathcal{L}^*} X_{\mathcal{L}}} })$.
\end{lem}
\begin{proof}
 By Lemma \ref{lem3}, to prove the achievability of $\mathcal{R} (\mathbb{A},{p_{V_{\mathcal{L}^*} U_{\mathcal{L}^*} X_{\mathcal{L}}} })$, it is sufficient to prove the achievability of $\mathcal{P} (g_{p_{V_{\mathcal{L}^*} U_{\mathcal{L}^*} X_{\mathcal{L}}} })$. By Lemma \ref{lem2}, to prove the achievability of $\mathcal{P} (g_{p_{V_{\mathcal{L}^*} U_{\mathcal{L}^*} X_{\mathcal{L}}} })$, it is sufficient to prove the achievability of the corner point $( I(V_{l};X_{l} | V_{1:l-1} X_{l_0}) + I(U_{l};X_{l} | U_{1:l-1}V_{l:L-1} X_{l_0}) )_{l \in \mathcal{L}^*}$, as the other corner points can be achieved similarly by relabelling the users and any point of the dominant face can be achieved by time-sharing between the corner points $\left\{ (C_{\pi(l)})_{l \in \mathcal{L}^*} :\pi \in \text{Sym}(L-1) \right\}$. Hence, since the constraints in \eqref{eqle} and \eqref{eql2a} are preserved under time sharing, the achievability of $\mathbf{R}^{\star}(\mathbb{A},{p_{V_{\mathcal{L}^*} U_{\mathcal{L}^*} X_{\mathcal{L}}}})$ implies the achievability of $\mathcal{R} (\mathbb{A},{p_{V_{\mathcal{L}^*} U_{\mathcal{L}^*} X_{\mathcal{L}}} })$.
\end{proof}

\subsection{Achieving $\mathbf{R}^{\star}(\mathbb{A},{p_{V_{\mathcal{L}^*} U_{\mathcal{L}^*} X_{\mathcal{L}}}})$} \label{seccod2} 

\subsubsection{Review of random binning and coding strategy overview}

We adopt a random binning approach, as it is convenient to precisely control the distribution of the random variables induced by our coding scheme to study the leakage~conditions. We will use the following results from \cite{yassaee2014achievability}. Specifically, Lemma \ref{lembinindep} provides a rate condition on the binning size required to perform privacy amplification \cite{bennett2002generalized}, Lemma \ref{lemSW} provides a rate condition on the binning size required to perform source coding with side information, and Lemma \ref{lemel} is key to showing that if a constraint holds on average over random binning choices, then there exists a specific binning that satisfies the constraint asymptotically.
\begin{lem}[{\cite[Th. 1]{yassaee2014achievability}}]\label{lembinindep}
Consider the joint distribution $p_{XZ}$ defined over the finite alphabet $ \mathcal{X} \times \mathcal{Z}$.  
Define a random binning, i.e., a  mapping 
$
{B} : \mathcal{X}^n \rightarrow \llbracket 1, 2^{nR} \rrbracket$ that  maps each sequence in $\mathcal{X}^n$ uniformly and independently to the set $\llbracket 1, 2^{nR} \rrbracket$. If $R < H(X| Z)$,
then $ \lim_{n\to \infty}
\mathbb{E} \left| p_{Z^n B(X^n)} - p_{Z^n}  p^{\textup{unif}} \right| = 0,
$ where the expectation is over the random binning choice and $p^{\textup{unif}}$ is the uniform distribution over $\llbracket 1, 2^{nR} \rrbracket$. 
\end{lem}

\begin{lem}[{\cite[Lemma 1]{yassaee2014achievability}}]\label{lemSW}
Consider a joint distribution $p_{XZ}$ and corresponding random binning $B$ as in Lemma \ref{lembinindep}, and consider the problem of reconstructing $X^n$ from $(B(X^n), Z^n)$. To this end, define a decoder $p^{\text{SW}}_{\hat{X} | Z^n B}$ similar to the achievability proof of Slepian-Wolf coding in \cite[Sec. 10.3.2]{el2011network}.  If $R > H(X| Z)$,
then $ \lim_{n\to \infty}
\mathbb{E} \mathbb{P}\left[ \hat{X}^n = X^n \right] = 0,
$ where the expectation is over the random binning choice.
\end{lem}

\begin{lem}[{\cite[Lemma 3]{yassaee2014achievability}}]\label{lemel}
Consider the joint distributions $p_{XY}$ and $q_{XY}$ defined over the finite alphabet $ \mathcal{X} \times \mathcal{Y}$. If $\mathbb{V}(p_{XY},q_{XY}) < \epsilon$, then there exists $x\in \mathcal{X}$ such that $\mathbb{V}(p_{Y|X=x},q_{Y|X=x}) < 2\epsilon$.
\end{lem}

 An overview of our coding strategy is as follows. First, we perform random binning of the sequences $U_{\mathcal{L}^*}^n$ and $V_{\mathcal{L}^*}^n$, generating associated bin indices $((M_l^V,J_l^V,M_l^U,J_l^U)_{l\in \mathcal{L}^*})$, and consider the joint distribution $p$ induced by this random binning. Based on this setup, we then construct a coding scheme satisfying the following properties:  
(i) the distribution $\tilde{p}$ induced by the coding scheme closely approximates $p$, provided the binning rates satisfy the conditions of Lemma~\ref{lembinindep},   
(ii) if each user $l\in \mathcal{L}^*$ transmits their bin index $M_l^V$, then the receiver can recover $V_{\mathcal{L}^*}^n$, assuming   $(J_l^V)_{l\in \mathcal{L}^*}$ is known at the receiver, and the binning rates meet the requirements of Lemma~\ref{lemSW}, (iii)  if each user $l\in \mathcal{L}^*$ transmits their bin index $M_l^U$, the receiver can recover $U_{\mathcal{L}^*}^n$, assuming now that $(J_l^U)_{l\in \mathcal{L}^*}$ and $V_{\mathcal{L}^*}^n$ are known at the receiver, and the binning rates meet the requirements of Lemma \ref{lemSW}. By construction of $U_{\mathcal{L}^*}^n$, the receiver can then estimate $F^n$ from $(U_{\mathcal{L}^*}^n,X^n_{l_0})$.
Finally, by Lemma \ref{lemel}, all three properties  (i), (ii), and (iii) continue to hold for a fixed choice of  $(J_l^V,J_l^U)_{l\in \mathcal{L}^*}$, which removes the need for shared randomness. The analysis of privacy leakage then leverages both the approximation of $p$ by $\tilde{p}$ and the properties of $p$.

\subsubsection{Coding scheme}
Fix $p_{V_{\mathcal{L}^*} U_{\mathcal{L}^*}  X_{\mathcal{L}}} = p_{X_{\mathcal{L}}} \prod_{l\in\mathcal{L}^*}  p_{U_l|X_l} p_{V_l|U_l}$ such that $H(F|U_{\mathcal{L}^*}X_{l_0})=0$, and consider $(V^n_{\mathcal{L}^*}, U^n_{\mathcal{L}^*}  ,X^n_{\mathcal{L}})$ independently and identically distributed according to $p_{V_{\mathcal{L}^*} U_{\mathcal{L}^*}  X_{\mathcal{L}}}$.

\textbf{Random binning}: 
\begin{itemize}
    \item 
For $l\in\mathcal{L}^*$, to each $v^n_l$ assign uniformly and independently two random bin indices $m_l^V \in \llbracket 1,2^{nR_l^V} \rrbracket$ and $j_l^V\in \llbracket 1,2^{n\tilde{R}_l^V} \rrbracket$. 
\item To each pair $(v^n_l,u^n_l)$ assign uniformly and independently two random bin indices $m_l^U \in \llbracket 1,2^{nR_l^U} \rrbracket$ and $j_l^U\in \llbracket 1,2^{n\tilde{R}_l^U} \rrbracket$. 
\item The distribution induced by this random binning is \begin{align*}&p_{V^n_{\mathcal{L}^*}U^n_{\mathcal{L}^*}X^n_{\mathcal{L}}M^V_{\mathcal{L}^*}M^U_{\mathcal{L}^*}J^V_{\mathcal{L}^*}J^U_{\mathcal{L}^*}} \\& \phantom{--}\triangleq p_{V^n_{\mathcal{L}^*} U^n_{\mathcal{L}^*}  X^n_{\mathcal{L}}} \prod_{l\in\mathcal{L}^*}p_{M_l^VJ_l^V|V^n_l}p_{M_l^UJ_l^U|V^n_lU^n_l}.\end{align*}
\end{itemize}

\textbf{Common-randomness-aided coding scheme}: 
\begin{itemize}
\item Assume that all the users have access to the common randomness $(J_l^V,J_l^U)_{l\in\mathcal{L}^*}$ uniformly distributed over $\bigtimes_{l\in\mathcal{L}^*} \llbracket 1, 2^{n\tilde{R}^V_l} \rrbracket \times \llbracket 1, 2^{n\tilde{R}^U_l} \rrbracket$.
\item User $l\in\mathcal{L}^*$ generates $V_l^n$ according to $p_{V^n_l|J_l^VX^n_l}$ and $U_l^n$ according to $ p_{U^n_l|J_l^UX^n_lV^n_l}$, then transmits $M^V_l$ generated according to $p_{M_l^V|V^n_l}$ and $M^U_l$ generated according to $p_{M_l^U|V^n_lU^n_l}$. 
\item By choosing for $l\in\mathcal{L}^*$, $\tilde{R}^V_l + R^V_l = H(V_l|V_{1:l-1}X_{l_0})+\epsilon/4$ and $\tilde{R}^U_l + R^U_l = H(U_l|U_{1:l-1}V_{\mathcal{L}^*}X_{l_0})+\epsilon/4$,  User $l_0$ can reconstruct with vanishing probability of error first $V^n_{\mathcal{L}^*}$ from $((M_l^V,J_l^V)_{l\in \mathcal{L}^*},X_{l_0}^n)$, and then $U^n_{\mathcal{L}^*}$ from $((M_l^U,J_l^U)_{l\in \mathcal{L}^*},V^n_{\mathcal{L}^*},X_{l_0}^n)$ with Lemma~\ref{lemSW}, i.e.,  source coding with side information.
\item Note that since $H(F|U_{\mathcal{L}^*}X_{l_0})=0$, there exists a deterministic function $\tilde{F}$ such that $\tilde{F}(u_{\mathcal{L}^*},x_{l_0})=F(x_{\mathcal{L}})$, for all $(u_{\mathcal{L}^*},x_{\mathcal{L}}) $ such that $p(u_{\mathcal{L}^*},x_{\mathcal{L}})>0$. Then,  User $l_0$ computes $\widehat{F}^n \triangleq \tilde{F}(\hat{U}_{\mathcal{L}^*}^n,X^n_{l_0})$. Since $(u^n_{\mathcal{L}^*},x^n_{\mathcal{L}})$ being typical sequences,  implies $ p((u_{\mathcal{L}^*})_i,(x_{\mathcal{L}})_i) \neq 0 ,\forall i \in \llbracket 1 , n\rrbracket$, we also have $\mathbb{P}[\widehat{F}^n \neq F^n] \xrightarrow{n \to \infty}0$. 
\item The distribution induced by the coding scheme is 
\begin{align*}
&\tilde{p}_{V^n_{\mathcal{L}^*}U^n_{\mathcal{L}^*}X^n_{\mathcal{L}}M^V_{\mathcal{L}^*}M^U_{\mathcal{L}^*}J^V_{\mathcal{L}^*}J^U_{\mathcal{L}^*}} \\
&\triangleq p_{ X^n_{\mathcal{L}}} \textstyle\prod_{l\in\mathcal{L}^*}  p_{M_l^V|V^n_l}p^{\text{unif}}_{J_l^V}p_{V^n_l|J_l^VX^n_l} \\
& \phantom{--} \times p_{M_l^U|V^n_lU^n_l}p^{\text{unif}}_{J_l^U} p_{U^n_l|J_l^UX^n_lV^n_l},\end{align*}
with $p^{\text{unif}}_{J_l^V}$ and $p^{\text{unif}}_{J_l^U}$ the uniform distributions over $\llbracket 1,2^{n\tilde{R}_l^V} \rrbracket$ and $\llbracket 1,2^{n\tilde{R}_l^U} \rrbracket$, $l\in\mathcal{L}^*$, respectively.

Note also that we have
\begin{align*}
&p_{V^n_{\mathcal{L}^*}U^n_{\mathcal{L}^*}X^n_{\mathcal{L}}M^V_{\mathcal{L}^*}M^U_{\mathcal{L}^*}J^V_{\mathcal{L}^*}J^U_{\mathcal{L}^*}}\\
& = p_{ X^n_{\mathcal{L}}} \textstyle\prod_{l\in\mathcal{L}^*}p_{U^n_l|X^n_l} p_{V^n_l|U^n_l} p_{M_l^VJ_l^V|V^n_l}p_{M_l^UJ_l^U|V^n_lU^n_l} \\
& = p_{ X^n_{\mathcal{L}}} \textstyle\prod_{l\in\mathcal{L}^*}p_{V^n_l|X^n_l} p_{U^n_l|V^n_lX^n_l} p_{M_l^V|V^n_l}p_{J_l^V|M_l^VV^n_l}\\
& \phantom{--------} \times p_{M_l^U|V^n_lU^n_l}p_{J_l^U|M_l^UV^n_lU^n_l} \\
& = p_{ X^n_{\mathcal{L}}} \textstyle\prod_{l\in\mathcal{L}^*}  p_{M_l^V|V^n_l}p_{J_l^VV^n_l|X^n_l}\\
& \phantom{--------} \times p_{M_l^U|V^n_lU^n_l}p_{J_l^UU^n_l|X^n_lV^n_l} \\
& = p_{ X^n_{\mathcal{L}}} \textstyle\prod_{l\in\mathcal{L}^*}  p_{M_l^V|V^n_l}p_{J_l^V|X^n_l}p_{V^n_l|J_l^VX^n_l}\\
& \phantom{--------} \times p_{M_l^U|V^n_lU^n_l}p_{J_l^U|X^n_lV^n_l} p_{U^n_l|J_l^UX^n_lV^n_l} ,
\end{align*}
such that by Lemma \ref{lembinindep} with $\tilde{R}^V_l = H(V_l|X_l) - \epsilon/4$ and $\tilde{R}^U_l = H(U_l|X_l V_l) - \epsilon/4$, $l\in\mathcal{L}^*$, $\epsilon >0$,   we have \begin{align} &\!\!\!\!\!\!\!\!\!\!\mathbb{E}[\mathbb{V}(p_{V^n_{\mathcal{L}^*}U^n_{\mathcal{L}^*}X^n_{\mathcal{L}}M^V_{\mathcal{L}^*}M^U_{\mathcal{L}^*}J^V_{\mathcal{L}^*}J^U_{\mathcal{L}^*}},\tilde{p}_{V^n_{\mathcal{L}^*}U^n_{\mathcal{L}^*}X^n_{\mathcal{L}}M^V_{\mathcal{L}^*}M^U_{\mathcal{L}^*}J^V_{\mathcal{L}^*}J^U_{\mathcal{L}^*}} \!)] \nonumber \\
& \xrightarrow{n\to \infty} 0,\label{eqind}\end{align} where the average is over the random binning choice. 
\end{itemize}

\textbf{Derandomization}: Finally, since the reliability constraint and the  distribution
induced by the coding scheme satisfies \eqref{eqind} on average  over the random binning choice, by Lemma \ref{lemel}, we conclude that there exist a specific binning and realizations of $(J^V_l,J^U_l)_{l\in\mathcal{L}^*}$ that preserve the reliability constraint and the constraint \eqref{eqind}. 

The analysis of the coding scheme is provided in Appendix~\ref{appanalysis}.

\section{Concluding Remarks} \label{secconcl}
We considered a function computation setting among multiple users where only a public and noiseless broadcast channel is available to the users. Our setting contrasts with function computation settings that do not consider information leakages, or make the assumption that secure communication links are available among users. We focused on studying optimal communication  and information leakage rates on the private user data for two models. In the first one, a fusion center needs to compute a function of the private user data. In the second one, there is no fusion center and a specific user must compute a function of the private data of all the users, including theirs.
For both settings, we derived a capacity region when the data of the users is independent. Additionally, we derived inner and outer regions for the capacity region for both settings when the data of the users is correlated. A possible extension of our model is to consider each user’s source as having private and public components, where privacy must be preserved only for the private component, similar to the setting studied in \cite{vamoua}.

\appendices

 \section{Proof of Theorem \ref{th5}} \label{Secth5}

Fix $\mathcal{A} \in \mathbb{A}$.  For $l\in\mathcal{L}$ and $j \in \llbracket 1, n \rrbracket$, we write $X_{l}^j \triangleq (X_{l,i})_{i\in \llbracket 1, j \rrbracket}$. And, for $\mathcal{S} \subseteq \mathcal{L}$, we define $(X_{\mathcal{S}})_i \triangleq (X_{l,i})_{l\in\mathcal{S}}$ and $(X_{\mathcal{S}})^{i-1} \triangleq (X_{l}^{i-1})_{l\in\mathcal{S}}$.  Then, we have 
\begin{align*}
&o(n)\\
& = H(F^n | X^n_{l_0}M_{\mathcal{L}^*})\\
& = \sum_{i=1}^n H(F_i | F^{i-1} X^n_{l_0}M_{\mathcal{L}^*})\\
& \stackrel{(a)}\geq  \sum_{i=1}^n H(F_i | (X_{\mathcal{L}^*})^{i-1} X^n_{l_0} M_{\mathcal{L}^*})\\
& = \sum_{i=1}^n H(F_i |(X_{\mathcal{A}})^{i-1} (X_{\mathcal{L}^*}\!)^{i-1} (X_{l_0})_{i+1}^n (X_{l_0}\!)_{i} (X_{l_0}\!)^{i-1} M_{\mathcal{L}^*} \!)\\
& \stackrel{(b)} = \sum_{i=1}^n H(F_i |(X_{\mathcal{A}})^{i-1} (X_{\mathcal{L}^*})^{i-1} (X_{l_0})_{i+1}^n (X_{l_0})_{i}  M_{\mathcal{L}^*} )\\
& \stackrel{(c)}=  \sum_{i=1}^n H(F_i | (X_{l_0})_i(U_{\mathcal{L}^*})_{i} )\\
& \stackrel{(d)} = n  H(F | X_{l_0} U_{\mathcal{L}^*} ), \numberthis \label{eqUb}
\end{align*}
where \begin{enumerate}[(a)]
    \item holds because $F^{i-1}$ is a function of $(X_{\mathcal{L}}\!)^{i-1}$; \item holds because $$\!\!(X_{l_0})^{i-1} \!- (X_{\mathcal{A}})^{i-1}\!,\!(X_{\mathcal{L}^*})^{i-1}\!, \!(X_{l_0})_{i+1}^n ,(X_{l_0})_{i} , M_{\mathcal{L}^*}\!\!) - F_i;$$
    \item and (d) hold with $$(U_l)_i \triangleq (M_l, X_{l}^{i-1}, (X_{l_0})_{i+1}^n , (X_{\mathcal{A}})^{i-1}),$$ $(U_{\mathcal{S}})_i \triangleq ((U_l)_i)_{l \in \mathcal{S}}$, $\mathcal{S} \subseteq \mathcal{L}^*$, and $U_l \triangleq (Q,U_{l,Q})$ with $Q$ uniformly distributed over $\llbracket 1 , n \rrbracket$ and independent of other random variables. 
    \end{enumerate}
Then, for $\mathcal{S} \subseteq \mathcal{L}^*$, we have
\begin{align*}
 R_{\mathcal{S}}  
&\geq \frac{1}{n} \sum_{l\in \mathcal{S}} H(M_l) \\
&\geq  \frac{1}{n}  H(M_{\mathcal{S}}) \\
& = \frac{1}{n}  I(M_{\mathcal{S}};X_{\mathcal{S}}^n)  \\
& \geq \frac{1}{n}  I(M_{\mathcal{S}};X_{\mathcal{S}}^n) - \frac{1}{n}  I(M_{\mathcal{S}};X_{l_0}^n) \numberthis \label{eql1b} \\
&\stackrel{(a)} = \frac{1}{n}  \smash{\sum_{i=1}^n} I(M_{\mathcal{S}};(X_{\mathcal{S}})_i | (X_{\mathcal{S}})^{i-1} (X_{l_0})_{i+1}^n)\\
&\phantom{------}- I(M_{\mathcal{S}};(X_{l_0})_i  | (X_{\mathcal{S}})^{i-1} (X_{l_0})_{i+1}^n )\\
& \stackrel{(b)}= \frac{1}{n}  \smash{\sum_{i=1}^n} I(M_{\mathcal{S}} (X_{\mathcal{S}})^{i-1} (X_{l_0})_{i+1}^n ;(X_{\mathcal{S}})_i  )\\
&\phantom{------}- I(M_{\mathcal{S}} (X_{\mathcal{S}})^{i-1} (X_{l_0})_{i+1}^n;(X_{l_0})_i   )\numberthis \label{eqrev2}\\
& \stackrel{(c)}= \frac{1}{n}  \smash{\sum_{i=1}^n} I(M_{\mathcal{S}} (X_{\mathcal{S}})^{i-1} (X_{l_0})_{i+1}^n  (X_{\mathcal{A}})^{i-1};(X_{\mathcal{S}})_i  )\\
&\phantom{------}- I(M_{\mathcal{S}} (X_{\mathcal{S}})^{i-1} (X_{l_0})_{i+1}^n (X_{\mathcal{A}})^{i-1} ;(X_{l_0})_i   )\\
&  \stackrel{(d)} =   \frac{1}{n}  \smash{\sum_{i=1}^n} I( (U_{\mathcal{S}})_i ;(X_{\mathcal{S}})_i |Q=i ) \\
&\phantom{------}- I( (U_{\mathcal{S}})_i ;(X_{l_0})_i |Q=i )\\
& =   I( (U_{\mathcal{S}})_Q ;(X_{\mathcal{S}})_Q |Q )-  I( (U_{\mathcal{S}})_Q ;(X_{l_0})_Q |Q )\\
& =   I( (U_{\mathcal{S}})_Q Q;(X_{\mathcal{S}})_Q )  - I( (U_{\mathcal{S}})_Q Q;(X_{l_0})_Q ) \\
&  \stackrel{(e)}=    I( U_{\mathcal{S}};X_{\mathcal{S}}) - I( U_{\mathcal{S}};X_{l_0}) \\
& \stackrel{(f)}   =   I(U_{\mathcal{S}}; X_{\mathcal{S}}| X_{l_0})\numberthis \label{eql2b}\\
& \stackrel{(g)}   =   I(U_{\mathcal{S}}V_{\mathcal{S}}; X_{\mathcal{S}}| X_{l_0})\\
& = I(  V_{\mathcal{S}}; X_{\mathcal{S}}| X_{l_0}) +I( U_{\mathcal{S}} ; X_{\mathcal{S}}| V_{\mathcal{S}} X_{l_0})  \\
& =  I( V_{\mathcal{S}} ; X_{\mathcal{S}} V_{\mathcal{S}^c}| X_{l_0}) - I(V_{\mathcal{S}};V_{\mathcal{S}^c}|X_{\mathcal{S}}X_{l_0})\\
&\phantom{--}+ I(U_{\mathcal{S}};X_{\mathcal{S}}U_{\mathcal{S}^c}|V_{\mathcal{S}}X_{l_0}) - I(U_{\mathcal{S}};U_{\mathcal{S}^c}|X_{\mathcal{S}} V_{\mathcal{S}}X_{l_0})  \\
& \geq  I( V_{\mathcal{S}} ; X_{\mathcal{S}} | V_{\mathcal{S}^c}X_{l_0}) - I(V_{\mathcal{S}};V_{\mathcal{S}^c}|X_{\mathcal{S}}X_{l_0})\\
&\phantom{--} + I(U_{\mathcal{S}};X_{\mathcal{S}}|U_{\mathcal{S}^c} V_{\mathcal{S}}X_{l_0}) - I(U_{\mathcal{S}};U_{\mathcal{S}^c}|X_{\mathcal{S}} V_{\mathcal{S}}X_{l_0})  ,
\end{align*}
where
\begin{enumerate}[(a)]
    \item holds by \cite[Lemma 4.1]{ahlswede1993common}; \item holds because $I((X_{\mathcal{S}})_i ; (X_{\mathcal{S}})^{i-1} (X_{l_0})_{i+1}^n)=0$ and  $I((X_{l_0})_i  ; (X_{\mathcal{S}})^{i-1} (X_{l_0})_{i+1}^n )=0$; \item holds because $ (X_{\mathcal{A} \backslash \mathcal{S}})^{i-1} - (M_{\mathcal{S}}, (X_{\mathcal{S}})^{i-1}, (X_{l_0})_{i+1}^n )- (X_{\mathcal{S}})_i$ and $ (X_{\mathcal{A\backslash \mathcal{S}}})^{i-1} - (M_{\mathcal{S}} ,(X_{\mathcal{S}})^{i-1} ,(X_{l_0})_{i+1}^n) - (X_{l_0})_i $ form Markov chains; 
    \item holds by definition of $(U_l)_i$, $i \in \llbracket 1 , n \rrbracket$, $l \in \mathcal{L}^*$; 
    \item holds by definition of  $U_l$ and $X_{\mathcal{S}} \triangleq (X_{\mathcal{S}})_Q $;
    \item holds because $U_{\mathcal{S}} - X_{\mathcal{S}} - X_{l_0}$; 
    \item holds  with $V_{l,i} \triangleq (M_l,  (X_{l_0})_{i+1}^n , (X_{\mathcal{A}})^{i-1})$ and $V_l \triangleq (Q,V_{l,Q})$.
\end{enumerate}

Then, we have
\begin{align*}
& \frac{1}{n}I(X_{\mathcal{A}^c}^n ; M_{\mathcal{L}^*} |  X_{\mathcal{A}}^n) \\
&= \frac{1}{n}I(X_{\mathcal{A}^c}^n ; M_{\mathcal{A}^{c*}} |  X_{\mathcal{A}}^n) + \frac{1}{n}I(X_{\mathcal{A}^c}^n ; M_{\mathcal{A}^*} |  M_{\mathcal{A}^{c*}} X_{\mathcal{A}}^n) \\
&= \frac{1}{n}I(X_{\mathcal{A}^c}^n ; M_{\mathcal{A}^{c*}} |  X_{\mathcal{A}}^n)  \\
&  \stackrel{(a)} = \frac{1}{n}I(X_{\mathcal{A}^c}^n ; M_{\mathcal{A}^{c*}}   ) - \frac{1}{n}I(    X_{\mathcal{A}}^n ;M_{\mathcal{A}^{c*}})\\
& \stackrel{(b)} = \frac{1}{n}I(X_{\mathcal{A}^{c*}}^n ; M_{\mathcal{A}^{c*}}   ) - \frac{1}{n}I(    X_{\mathcal{A}}^n ;M_{\mathcal{A}^{c*}})\\
&  = \frac{1}{n}I(X_{\mathcal{A}^{c*}}^n ; M_{\mathcal{A}^{c*}}   )- \frac{1}{n}I(X_{l_0}^n ; M_{\mathcal{A}^{c*}}   ) + \frac{1}{n}I(X_{l_0}^n ; M_{\mathcal{A}^{c*}}   ) \\
& \phantom{--}- \frac{1}{n}I(    X_{\mathcal{A}}^n ;M_{\mathcal{A}^{c*}})\\
& \stackrel{(c)} = I( U_{\mathcal{A}^{c*}};X_{\mathcal{A}^{c*}}|X_{l_0}) + \frac{1}{n}I(X_{l_0}^n ; M_{\mathcal{A}^{c*}}   ) - \frac{1}{n}I(    X_{\mathcal{A}}^n ;M_{\mathcal{A}^{c*}})\\
& \stackrel{(d)} = I( U_{\mathcal{A}^{c*}};X_{\mathcal{A}^{c*}}|X_{l_0}) \\
& \phantom{--} + \frac{1}{n} \smash{\sum_{i=1}^n} I( (X_{l_0})_i ;M_{\mathcal{A}^{c*}} |(X_{\mathcal{A}})^{i-1} (X_{l_0})_{i+1}^n)\\
& \phantom{-------}-I( (X_{\mathcal{A}})_i ;M_{\mathcal{A}^{c*}} |(X_{\mathcal{A}})^{i-1} (X_{l_0})_{i+1}^n) \\
&  = I( U_{\mathcal{A}^{c*}};X_{\mathcal{A}^{c*}}|X_{l_0}) \\
& \phantom{--} + \frac{1}{n} \smash{\sum_{i=1}^n} I( (X_{l_0})_i; (X_{\mathcal{A}})^{i-1} (X_{l_0})_{i+1}^n M_{\mathcal{A}^{c*}} )\\
& \phantom{-------}-I( (X_{\mathcal{A}})_i ;(X_{\mathcal{A}})^{i-1} (X_{l_0})_{i+1}^n M_{\mathcal{A}^{c*}}  ) \\
& \stackrel{(e)} = I( U_{\mathcal{A}^{c*}};X_{\mathcal{A}^{c*}}|X_{l_0}) \\
& \phantom{--}+ \frac{1}{n} \sum_{i=1}^n I( (X_{l_0})_i; (V_{\mathcal{A}^{c*}})_i )-I( (X_{\mathcal{A}})_i ;(V_{\mathcal{A}^{c*}})_i )\\
& = I( U_{\mathcal{A}^{c*}};X_{\mathcal{A}^{c*}}|X_{l_0}) + I( X_{l_0}; V_{\mathcal{A}^{c*}} )-I( X_{\mathcal{A}} ;V_{\mathcal{A}^{c*}} )\\
& \stackrel{(f)} = I( U_{\mathcal{A}^{c*}};X_{\mathcal{A}^{c*}}) - I( U_{\mathcal{A}^{c*}};X_{l_0}) \\
& \phantom{--}+ I( X_{l_0};U_{\mathcal{A}^{c*}} V_{\mathcal{A}^{c*}} )- I( X_{l_0};U_{\mathcal{A}^{c*}}| V_{\mathcal{A}^{c*}} )\\
& \phantom{--}-I( X_{\mathcal{A}} ;U_{\mathcal{A}^{c*}}V_{\mathcal{A}^{c*}} )  +I( X_{\mathcal{A}} ;U_{\mathcal{A}^{c*}}|V_{\mathcal{A}^{c*}} )\\
& \stackrel{(g)}= I( U_{\mathcal{A}^{c*}};X_{\mathcal{A}^{c*}}|X_{\mathcal{A}}) + I( X_{l_0}; V_{\mathcal{A}^{c*}} |U_{\mathcal{A}^{c*}})\\
& \phantom{--}- I( X_{l_0};U_{\mathcal{A}^{c*}}| V_{\mathcal{A}^{c*}} )-I( X_{\mathcal{A}} ;V_{\mathcal{A}^{c*}}|U_{\mathcal{A}^{c*}} )  \\
& \phantom{--}+I( X_{\mathcal{A}} ;U_{\mathcal{A}^{c*}}|V_{\mathcal{A}^{c*}} )\\
& \stackrel{(h)}= I( U_{\mathcal{A}^{c*}};X_{\mathcal{A}^{c*}}|X_{\mathcal{A}}) \\
& \phantom{--} - I( X_{l_0};U_{\mathcal{A}^{c*}}| V_{\mathcal{A}^{c*}} ) +I( X_{\mathcal{A}} ;U_{\mathcal{A}^{c*}}|V_{\mathcal{A}^{c*}} )\\
&\geq I( U_{\mathcal{A}^{c*}};X_{\mathcal{A}^{c*}}|X_{\mathcal{A}}) \\
& \phantom{--} - [I( X_{l_0};U_{\mathcal{A}^{c*}}| V_{\mathcal{A}^{c*}} ) -I( X_{\mathcal{A}} ;U_{\mathcal{A}^{c*}}|V_{\mathcal{A}^{c*}} )]^+, \numberthis \label{eqdeltaa}
\end{align*}
where 
\begin{enumerate}[(a)]
\item  holds because $ M_{\mathcal{A}^{c*}} - X_{\mathcal{A}^c}^n  -  X_{\mathcal{A}}^n$;
\item holds by considering the two cases $l_0 \in \mathcal{A}$ and $l_0 \notin \mathcal{A}$, and because $I(X_{l_0}^n;M_{\mathcal{A}^{c*}}|X^n_{\mathcal{A}^{c*}})=0$; 
\item holds by the steps between \eqref{eql1b} and \eqref{eql2b}; 
\item holds by \cite[Lemma 4.1]{ahlswede1993common}; \item holds by the definition of $V_{l,i}$, $l\in\mathcal{L}^*$, $i\in \llbracket 1 , n \rrbracket$; \item holds by the chain rule and because $X_{l_0} - X_{\mathcal{A}^{c*}} - U_{\mathcal{A}^{c*}}$; \item holds by the chain rule and because $X_{\mathcal{A}} - X_{\mathcal{A}^{c*}} - U_{\mathcal{A}^{c*}}$; \item holds because $I( X_{l_0}; V_{\mathcal{A}^{c*}} |U_{\mathcal{A}^{c*}})=0=I( X_{\mathcal{A}} ;V_{\mathcal{A}^{c*}}|U_{\mathcal{A}^{c*}} )$ by using the definitions of $V_{\mathcal{A}^{c*}}$ and $U_{\mathcal{A}^{c*}}$.
\end{enumerate}

Finally, if $\mathcal{A} \ni l_0$, then we have
\begin{align*}
\Delta_{\mathcal{A}}
& \geq \frac{1}{n}I(X_{\mathcal{A}^c}^n ; M_{\mathcal{L}^*} | F^n X_{\mathcal{A}}^n)  \\
& = \frac{1}{n}I(X_{\mathcal{A}^c}^n ; M_{\mathcal{L}^*} F^n| X_{\mathcal{A}}^n) -  I( X_{\mathcal{A}^c} ;F | X_{\mathcal{A}})\\
& = \frac{1}{n}I(X_{\mathcal{A}^c}^n ; M_{\mathcal{L}^*} | X_{\mathcal{A}}^n)+\frac{1}{n}I(X_{\mathcal{A}^c}^n ;  F^n| M_{\mathcal{L}^*} X_{\mathcal{A}}^n) \\
& \phantom{--}-  I( X_{\mathcal{A}^c} ;F | X_{\mathcal{A}})\\
& \stackrel{(a)}  = \frac{1}{n}I(X_{\mathcal{A}^c}^n ; M_{\mathcal{L}^*} | X_{\mathcal{A}}^n)  -  I( X_{\mathcal{A}^c} ;F | X_{\mathcal{A}}) + o(1)\\
& \stackrel{(b)} \geq  I( U_{\mathcal{A}^{c}};X_{\mathcal{A}^{c}}|X_{\mathcal{A}})    -  I( X_{\mathcal{A}^c} ;F | X_{\mathcal{A}}) + o(1),
\end{align*}
where \begin{enumerate}[(a)]
\item holds by Fano's inequality; \item holds by \eqref{eqdeltaa} and because $l_0 \in \mathcal{A}$.
\end{enumerate}

\section{Proofs of the lemmas of Section \ref{secreduction2}} \label{app_lem}

\subsection{Proof of Lemma \ref{lem1}}
Clearly, $g_{p_{V_{\mathcal{L}^*} U_{\mathcal{L}^*}  X_{\mathcal{L}}}}$ is normalized. $g_{p_{V_{\mathcal{L}^*} U_{\mathcal{L}^*}  X_{\mathcal{L}}}}$ is also non-decreasing because for $\mathcal{S}_1,\mathcal{S}_2 \subseteq \mathcal{L}^*$ such that  $\mathcal{S}_1\subseteq\mathcal{S}_2$, we have
\begin{align*}
&g_{p_{V_{\mathcal{L}^*} U_{\mathcal{L}^*}  X_{\mathcal{L}}}} (\mathcal{S}_2) \\
& = I(V_{\mathcal{S}_2};X_{\mathcal{L}} | V_{\mathcal{S}_2^c} X_{l_0}) +  I(U_{\mathcal{S}_2};X_{\mathcal{L}} | U_{\mathcal{S}_2^c} V_{\mathcal{S}_2} X_{l_0}) \\
& = I(V_{\mathcal{S}_2 \backslash \mathcal{S}_1};X_{\mathcal{L}} | V_{\mathcal{S}_2^c}X_{l_0})  + I(V_{\mathcal{S}_1};X_{\mathcal{L}} | V_{\mathcal{S}_2^c}V_{\mathcal{S}_2 \backslash \mathcal{S}_1} X_{l_0}) \\
& \phantom{--} + I(U_{\mathcal{S}_2 \backslash \mathcal{S}_1};X_{\mathcal{L}} | U_{\mathcal{S}_2^c}V_{\mathcal{S}_2}X_{l_0})  \\
& \phantom{--}+ I(U_{\mathcal{S}_1};X_{\mathcal{L}} | U_{\mathcal{S}_2^c}U_{\mathcal{S}_2 \backslash \mathcal{S}_1}V_{\mathcal{S}_2} X_{l_0}) \\
& \geq  I(V_{\mathcal{S}_1};X_{\mathcal{L}} | V_{\mathcal{S}_2^c}V_{\mathcal{S}_2 \backslash \mathcal{S}_1} X_{l_0}) \\
& \phantom{--}+ I(U_{\mathcal{S}_1};X_{\mathcal{L}} | U_{\mathcal{S}_2^c}U_{\mathcal{S}_2 \backslash \mathcal{S}_1}V_{\mathcal{S}_2} X_{l_0}) \\
&  \stackrel{(a)} =  I(V_{\mathcal{S}_1};X_{\mathcal{L}} | V_{\mathcal{S}_1^c}  X_{l_0}) + I(U_{\mathcal{S}_1};X_{\mathcal{L}} | U_{\mathcal{S}_1^c} V_{\mathcal{S}_2} X_{l_0}) \\
& =  I(V_{\mathcal{S}_1};X_{\mathcal{L}} | V_{\mathcal{S}_1^c}  X_{l_0}) + I(U_{\mathcal{S}_1} V_{\mathcal{S}_2 \backslash \mathcal{S}_1}; X_{\mathcal{L}} | U_{\mathcal{S}_1^c} V_{\mathcal{S}_1} X_{l_0})\\
& \phantom{--} - I(  V_{\mathcal{S}_2 \backslash \mathcal{S}_1}; X_{\mathcal{L}}|  U_{\mathcal{S}_1^c}  V_{\mathcal{S}_1} X_{l_0}) \\
& \stackrel{(b)} =  I(V_{\mathcal{S}_1};X_{\mathcal{L}} | V_{\mathcal{S}_1^c}  X_{l_0}) + I(U_{\mathcal{S}_1} V_{\mathcal{S}_2 \backslash \mathcal{S}_1}; X_{\mathcal{L}} | U_{\mathcal{S}_1^c} V_{\mathcal{S}_1} X_{l_0}) \\
& \geq  I(V_{\mathcal{S}_1};X_{\mathcal{L}} | V_{\mathcal{S}_1^c} X_{l_0}) +  I(U_{\mathcal{S}_1};X_{\mathcal{L}} | U_{\mathcal{S}_1^c} V_{\mathcal{S}_1} X_{l_0})\\
& = g_{p_{V_{\mathcal{L}^*} U_{\mathcal{L}^*}  X_{\mathcal{L}}}} (\mathcal{S}_1),
\end{align*}
where \begin{enumerate}[(a)]
\item holds because $\mathcal{S}_1\subseteq\mathcal{S}_2$; \item holds because \begin{align*}&I(  V_{\mathcal{S}_2 \backslash \mathcal{S}_1}; X_{\mathcal{L}}|  U_{\mathcal{S}_1^c}  V_{\mathcal{S}_1} X_{l_0}) \\&\leq I(V_{\mathcal{S}_2 \backslash \mathcal{S}_1}; X_{\mathcal{L}}  V_{\mathcal{S}_1}  U_{\mathcal{S}_2^c} |  U_{\mathcal{S}_2\backslash \mathcal{S}_1}   ) = 0.\end{align*}
\end{enumerate}
We now show that $g_{p_{V_{\mathcal{L}^*} U_{\mathcal{L}^*} X_{\mathcal{L}}}}$ is supermodular by proving that $h_1 : \mathcal{S} \mapsto I(V_{\mathcal{S}};X_{\mathcal{L}} | V_{\mathcal{S}^c} X_{l_0})$ and $h_2 : \mathcal{S} \mapsto I(U_{\mathcal{S}};X_{\mathcal{L}} | U_{\mathcal{S}^c} V_{\mathcal{S}} X_{l_0})$ are supermodular.  Note that for any $\mathcal{S} \subseteq \mathcal{L}^*$, we have $h_1(\mathcal{S}) = I(V_{\mathcal{L}}; X_{\mathcal{L}}|X_{l_0}) - I(V_{\mathcal{S}^c};X_{\mathcal{L}}|X_{l_0})$, hence, since for any submodular function $e$, the function $\mathcal{S} \mapsto - e(\mathcal{S}^c)$ is supermodular, it is sufficient to prove submodularity of $h: \mathcal{S}  \mapsto I(V_{\mathcal{S}};X_{\mathcal{L}}|X_{l_0})$ to prove that $h_1$ is supermodular. For $\mathcal{S}, \mathcal{T} \subseteq \mathcal{L}^*$, we have
\begin{align*}
&h(\mathcal{S} \cup \mathcal{T}) + h(\mathcal{S} \cap \mathcal{T}) \\
& =  I(V_{\mathcal{S} \cup \mathcal{T}};X_{\mathcal{L}}|X_{l_0}) + I(V_{\mathcal{S}\cap \mathcal{T}};X_{\mathcal{L}}|X_{l_0}) \\
& \stackrel{(a)} =  I(V_{\mathcal{S} };X_{\mathcal{L}}|X_{l_0}) + I(V_{\mathcal{T} \backslash \mathcal{S}};X_{\mathcal{L}}|X_{l_0} V_{\mathcal{S}}) \\
&\phantom{--}+ I(V_{\mathcal{T} };X_{\mathcal{L}}|X_{l_0}) - I(V_{\mathcal{T}\backslash \mathcal{S}};X_{\mathcal{L}}|X_{l_0} V_{\mathcal{S} \cap \mathcal{T}}) \\
& \stackrel{(b)} =  I(V_{\mathcal{S} };X_{\mathcal{L}}|X_{l_0}) + H(V_{\mathcal{T} \backslash \mathcal{S}}|X_{l_0}V_{\mathcal{S}})  \\
&\phantom{--}+ I(V_{\mathcal{T} };X_{\mathcal{L}}|X_{l_0}) - H(V_{\mathcal{T}\backslash \mathcal{S}} |X_{l_0}V_{\mathcal{S} \cap \mathcal{T}}) \\
& \stackrel{(c)} \leq I(V_{\mathcal{S} };X_{\mathcal{L}}|X_{l_0}) +   I(V_{\mathcal{T} };X_{\mathcal{L}}|X_{l_0})\\
& = h(\mathcal{S} ) + h(  \mathcal{T}) ,
\end{align*}
where \begin{enumerate}[(a)]
\item  holds by the chain rule applied twice; \item holds because \begin{align*}
&    I(V_{\mathcal{T} \backslash \mathcal{S}};X_{\mathcal{L}}|X_{l_0}V_{\mathcal{S}})  - I(V_{\mathcal{T}\backslash \mathcal{S}};X_{\mathcal{L}}|X_{l_0} V_{\mathcal{S} \cap \mathcal{T}})\\& = H(V_{\mathcal{T} \backslash \mathcal{S}}|X_{l_0} V_{\mathcal{S}}) - H(V_{\mathcal{T}\backslash \mathcal{S}} | X_{l_0} V_{\mathcal{S} \cap \mathcal{T}})\end{align*} by the Markov chains $V_{\mathcal{T} \backslash \mathcal{S}} - X_{\mathcal{T} \backslash \mathcal{S}} - (V_{\mathcal{S}},X_{\mathcal{L} \backslash (\mathcal{T} \backslash \mathcal{S}) },X_{l_0})$ and $V_{\mathcal{T} \backslash \mathcal{S}} - X_{\mathcal{T} \backslash \mathcal{S}} - (V_{\mathcal{S}\cap \mathcal{T}},X_{\mathcal{L} \backslash (\mathcal{T} \backslash \mathcal{S}) },X_{l_0})$; \item holds because conditioning reduces entropy. 
\end{enumerate}

Then, for any $\mathcal{S} \subseteq \mathcal{L}^*$, we have 
\begin{align*}
h_2(\mathcal{S}) 
& = I(U_{\mathcal{S}};X_{\mathcal{L}} | U_{\mathcal{S}^c} V_{\mathcal{S}} X_{l_0})\\
& = I(U_{\mathcal{S}}V_{\mathcal{S}};X_{\mathcal{L}} | U_{\mathcal{S}^c}  X_{l_0}) - I(V_{\mathcal{S}};X_{\mathcal{L}} | U_{\mathcal{S}^c} X_{l_0}) \\
& = I(U_{\mathcal{S}}V_{\mathcal{L}};X_{\mathcal{L}} | U_{\mathcal{S}^c}  X_{l_0}) - I(V_{\mathcal{L}};X_{\mathcal{L}} | U_{\mathcal{S}^c} X_{l_0}) \\
& = I(U_{\mathcal{S}};X_{\mathcal{L}} | U_{\mathcal{S}^c} V_{\mathcal{L}} X_{l_0}) \\
& = I(U_{\mathcal{L}};X_{\mathcal{L}} |  V_{\mathcal{L}} X_{l_0}) - I(U_{\mathcal{S}^c};X_{\mathcal{L}} | V_{\mathcal{L}} X_{l_0}),
\end{align*}
where the third equality holds by the chain rule and the Markov chains $V_{\mathcal{S}^c} - (U_{\mathcal{S}},V_{\mathcal{S}},U_{\mathcal{S}^c},X_{l_0}) - X_{\mathcal{L}}$ and $V_{\mathcal{S}^c} - (V_{\mathcal{S}},U_{\mathcal{S}^c},X_{l_0}) - X_{\mathcal{L}}$. Finally, we conclude that, similar to $h_1$, $h_2$ is supermodular.

\subsection{Proof of Lemma \ref{lem2}}
For convenience, we write $g$ instead of $g_{p_{V_{\mathcal{L}^*} U_{\mathcal{L}^*}  X_{\mathcal{L}}}}$ in the following. $(i)$ follows from Lemma~\ref{lem1} and Definition   \ref{defcp}. $(ii)$ follows from $(i)$ and Lemma \ref{lempm}. Then, by $(i)$ and Lemma \ref{lempm}, we also have 
\eqref{eqdom} with $\forall \pi \in \text{Sym}(L-1), \forall l \in \mathcal{L}^*,$
$$C_{\pi(l)} \triangleq  g\left( \pi (l:L-1) \right) - g \left( \pi(l+1:L-1) \right), $$
which one can rewrite as 
\begin{align*}
&C_{\pi(l)} \\
& = I(V_{\pi(l:L-1)};X_{\mathcal{L}} | V_{(\pi (l:L-1))^c} X_{l_0}) \\
& \phantom{-}
 - I(V_{\pi(l+1:L-1)};X_{\mathcal{L}} | V_{(\pi(l+1:L-1))^c} X_{l_0}) \\
& \phantom{-} +  I(U_{\pi(l:L-1)};X_{\mathcal{L}} | U_{(\pi (l:L-1))^c} V_{\pi(l:L-1)} X_{l_0}) \\
& \phantom{-}
 - I(U_{\pi(l+1:L-1)};X_{\mathcal{L}} | U_{(\pi(l+1:L-1))^c} V_{\pi(l+1:L-1)} X_{l_0}) \\
&\stackrel{(a)} = I(V_{\pi (l)};X_{\mathcal{L}} | V_{\pi(1:l-1)} X_{l_0}) 
  \\
& \phantom{-} +  I(U_{\pi(l:L-1)};X_{\mathcal{L}} | U_{(\pi(l:L-1))^c} V_{\pi(l:L-1)} X_{l_0}) \\
& \phantom{-}
 - I(U_{\pi(l+1:L-1)}V_{\pi (l)};X_{\mathcal{L}} | U_{(\pi(l+1:L-1))^c} V_{\pi(l+1:L-1)} X_{l_0})   \\
& \phantom{-} + I(V_{\pi (l)};X_{\mathcal{L}} | U_{\pi(1:L-1)} V_{\pi(l+1:L-1)} X_{l_0})\\
&\stackrel{(b)} = I(V_{\pi (l)};X_{\mathcal{L}} | V_{\pi(1:l-1)} X_{l_0}) 
  \\
& \phantom{-} +  I(U_{\pi(l:L-1)};X_{\mathcal{L}} | U_{(\pi (l:L-1))^c} V_{\pi(l:L-1)} X_{l_0}) 
\\
& \phantom{-} - I(U_{\pi(l+1:L-1)}V_{\pi (l)};X_{\mathcal{L}} | U_{(\pi(l+1:L-1))^c} V_{\pi(l+1:L-1)} X_{l_0})   \\
&\stackrel{(c)} = I(V_{\pi (l)};X_{\mathcal{L}} | V_{\pi (1:l-1)} X_{l_0}) 
  \\
& \phantom{-} +  I(U_{\pi(l:L-1)};X_{\mathcal{L}} | U_{(\pi(l:L-1))^c} V_{\pi(l:L-1)} X_{l_0}) 
\\
& \phantom{-} - I(U_{\pi(l+1:L-1)};X_{\mathcal{L}} | U_{(\pi(l+1:L-1))^c} V_{\pi(l:L-1)} X_{l_0})   \\
& \phantom{-} - I(V_{\pi (l)};X_{\mathcal{L}} | U_{\pi(1:l)} V_{\pi(l+1:L-1)} X_{l_0})\\
& \stackrel{(d)} = I(V_{\pi (l)};X_{\mathcal{L}} | V_{\pi (1:l-1)} X_{l_0}) \\
& \phantom{-}+  I(U_{\pi (l)};X_{\mathcal{L}} | U_{(\pi (l:L-1))^c} V_{\pi (l:L-1)} X_{l_0})  \\
& \stackrel{(e)} = I(V_{\pi (l)};X_{\pi(l)} | V_{\pi (1:l-1)}X_{l_0}) \\
& \phantom{-}+  I(V_{\pi (l)}; X_{(\pi(l))^c}| X_{\pi(l)} V_{\pi (1:l-1)}X_{l_0}) \\
& \phantom{-} +I(U_{\pi (l)};X_{\pi(l)} | U_{\pi (1:l-1)} V_{\pi (l:L-1)} X_{l_0}) \\
& \phantom{-}+  I(U_{\pi (l)}; X_{(\pi(l))^c}| X_{\pi(l)} U_{\pi(1:l-1)}V_{\pi (l:L-1)} X_{l_0})\\
& \stackrel{(f)} = I(V_{\pi (l)};X_{\pi(l)} | V_{\pi(1:l-1)}X_{l_0}) \\
& \phantom{-}+  I(U_{\pi (l)};X_{\pi(l)} | U_{\pi (1:l-1)} V_{\pi (l:L-1)} X_{l_0}),
\end{align*}
where
\begin{enumerate}[(a)]
    \item\!\!, (c), and (e) hold by the chain rule;
    \item holds because \begin{align*} &I(V_{\pi (l)};X_{\mathcal{L}} | U_{\pi(1:L-1)} V_{\pi(l+1:L-1)} X_{l_0}) \\&\leq  I(V_{\pi (l)};X_{\mathcal{L}} U_{\pi(1:L-1)} V_{\pi(l+1:L-1)} | U_{\pi(l)}  ) = 0;\end{align*} \addtocounter{enumi}{1}\item holds by the chain rule and because \begin{align*} &I(V_{\pi (l)};X_{\mathcal{L}} | U_{\pi(1:l)} V_{\pi(l+1:L-1)} X_{l_0})\\& \leq I(V_{\pi (l)};X_{\mathcal{L}} U_{\pi(1:l-1)} V_{\pi(l+1:L-1)}| U_{\pi(l)}  ) = 0;\end{align*}\addtocounter{enumi}{1} 
    \item holds because \begin{align*} &I(V_{\pi (l)}; X_{(\pi(l))^c}| X_{\pi(l)} V_{\pi(1:l-1)}X_{l_0}) \\&\leq I(V_{\pi (l)}; X_{(\pi(l))^c}V_{\pi (1:l-1)}| X_{\pi(l)} X_{l_0}) = 0;\end{align*} and \begin{align*} &\!\! I(U_{\pi (l)}; X_{(\pi(l))^c}| X_{\pi(l)} U_{\pi (1:l-1)}V_{\pi (l:L-1)} X_{l_0}) \\& \!\! \leq   I(U_{\pi (l)}V_{\pi (l)}; X_{(\pi(l))^c} U_{\pi (1:l-1)}V_{\pi (l+1:L-1)} | X_{\pi(l)}   ) =0.\end{align*}
\end{enumerate}
\subsection{Proof of Lemma \ref{lem3}}
By the definitions of $
\mathcal{R} ({p_{V_{\mathcal{L}^*} U_{\mathcal{L}^*}  X_{\mathcal{L}}} })$ and $\mathcal{P} (g_{p_{V_{\mathcal{L}^*} U_{\mathcal{L}^*}  X_{\mathcal{L}}} })$, it is sufficient to prove that $g_{p_{V_{\mathcal{L}^*} U_{\mathcal{L}^*}  X_{\mathcal{L}}} } = \bar{g}_{p_{V_{\mathcal{L}^*} U_{\mathcal{L}^*}  X_{\mathcal{L}}} }$. For $\mathcal{S} \subseteq \mathcal{L}^*$, we have
\begin{align*}
&g_{p_{V_{\mathcal{L}^*} U_{\mathcal{L}^*}  X_{\mathcal{L}}} } (\mathcal{S}) \\
&= I(V_{\mathcal{S}};X_{\mathcal{L}} | V_{\mathcal{S}^c} X_{l_0}) +  I(U_{\mathcal{S}};X_{\mathcal{L}} | U_{\mathcal{S}^c} V_{\mathcal{S}} X_{l_0}) \\
& = I(V_{\mathcal{S}};X_{\mathcal{S}} | V_{\mathcal{S}^c}X_{l_0}) + I(V_{\mathcal{S}};X_{\mathcal{S}^c} | X_{\mathcal{S}} V_{\mathcal{S}^c}X_{l_0}) \\
&\phantom{--}+ I(U_{\mathcal{S}};X_{\mathcal{S}} | U_{\mathcal{S}^c}V_{\mathcal{S}}X_{l_0}) + I(U_{\mathcal{S}};X_{\mathcal{S}^c} | X_{\mathcal{S}} U_{\mathcal{S}^c}V_{\mathcal{S}}X_{l_0}) \\
&  \stackrel{(a)}= I(V_{\mathcal{S}};X_{\mathcal{S}} | V_{\mathcal{S}^c}X_{l_0}) +  I(U_{\mathcal{S}};X_{\mathcal{S}} | U_{\mathcal{S}^c}V_{\mathcal{S}}X_{l_0}) \\
&  \stackrel{(b)}  = I(V_{\mathcal{S}};X_{\mathcal{S}} | V_{\mathcal{S}^c}X_{l_0})  - I(V_{\mathcal{S}};V_{\mathcal{S}^c}|X_{\mathcal{S}} X_{l_0}) \\
&\phantom{--}+  I(U_{\mathcal{S}};X_{\mathcal{S}} | U_{\mathcal{S}^c}V_{\mathcal{S}}X_{l_0}) -  I(U_{\mathcal{S}};U_{\mathcal{S}^c} | X_{\mathcal{S}}V_{\mathcal{S}}X_{l_0})  \\
& = \bar{g}_{p_{V_{\mathcal{L}^*} U_{\mathcal{L}^*}  X_{\mathcal{L}}} } (\mathcal{S}),
\end{align*}
where \begin{enumerate}[(a)]
\item  holds because of the Markov chains $V_{\mathcal{S}} - ( X_{\mathcal{S}}, V_{\mathcal{S}^c}, X_{l_0})- X_{\mathcal{S}^c}$  and $U_{\mathcal{S}} - ( X_{\mathcal{S}}, U_{\mathcal{S}^c},V_{\mathcal{S}}, X_{l_0})- X_{\mathcal{S}^c}$; \item holds because of the Markov chains $V_{\mathcal{S}}- (X_{\mathcal{S}}, X_{l_0}) - V_{\mathcal{S}^c}$  and $U_{\mathcal{S}}- (X_{\mathcal{S}},V_{\mathcal{S}}, X_{l_0}) - U_{\mathcal{S}^c}$. 
\end{enumerate}

\section{Analysis of the coding scheme in Theorem \ref{thach2}} \label{appanalysis}
\subsection{Communication rates} The communication rate for User $l\in\mathcal{L}^*$ is 
\begin{align*}
&R_l^V  + R_l^U \\
& = H(V_l|V_{1:l-1}X_{l_0}) - H(V_l|X_l)\\
&\phantom{--} + H(U_l|U_{1:l-1} V_{\mathcal{L}^*}X_{l_0}) - H(U_l|X_l V_l)+\epsilon\\
& \stackrel{(a)}= H(V_l|V_{1:l-1}X_{l_0}) - H(V_l|X_lV_{1:l-1}X_{l_0}) \\
&\phantom{--}+ H(U_l|U_{1:l-1} V_{\mathcal{L}^*}X_{l_0}) - H(U_l|X_l U_{1:l-1} V_{\mathcal{L}^*}X_{l_0})+\epsilon\\
& = I(V_l;X_l|V_{1:l-1}X_{l_0})  \\
&\phantom{--}+ I(U_l;X_l|U_{1:l-1} V_{\mathcal{L}^*}X_{l_0})  +\epsilon \numberthis \label{eqRU}\\
& \stackrel{(b)} = I(V_l;X_l|V_{1:l-1}X_{l_0})  + I(V_{1:l-1};X_{l} | U_{1:l-1}V_{l:L-1} X_{l_0})\\
&\phantom{--} +I(U_{l};X_{l} | U_{1:l-1}V_{1:L-1} X_{l_0})  +\epsilon\\
& \stackrel{(c)} = I(V_l;X_l|V_{1:l-1}X_{l_0}) \\
&\phantom{--} + I(U_{l}V_{1:l-1};X_{l} | U_{1:l-1}V_{l:L-1} X_{l_0}) +\epsilon\\
& \stackrel{(d)}  = I(V_l;X_l|V_{1:l-1}X_{l_0})  + I(U_{l}V_{1:l-1};X_{l} | U_{1:l-1}V_{l:L-1} X_{l_0})\\
&\phantom{--} -I(V_{1:l-1};X_{l} | U_{1:l}V_{l:L-1} X_{l_0}) +\epsilon\\
& \stackrel{(e)} = I(V_l;X_l|V_{1:l-1}X_{l_0})  + I(U_{l};X_{l} | U_{1:l-1}V_{l:L-1} X_{l_0})+\epsilon,
\end{align*}
where \begin{enumerate}[(a)]
\item  holds by the Markov chains $V_l-X_l-(V_{1:l-1},X_{l_0})$ and $U_l- (X_l,V_l) -( U_{1:l-1} ,V_{\mathcal{L}^*},X_{l_0})$; \item  holds because \begin{align*}&I(V_{1:l-1};X_{l} | U_{1:l-1}V_{l:L-1} X_{l_0}) \\&\leq I(V_{1:l-1};X_{l}V_{l:L-1} X_{l_0} | U_{1:l-1}) = 0;\end{align*} \item and (e)  hold by the chain rule; \item holds because \begin{align*}&I(V_{1:l-1};X_{l} | U_{1:l}V_{l:L-1} X_{l_0}) \\& \leq I(V_{1:l-1};X_{l} U_l V_{l:L-1} X_{l_0} | U_{1:l-1})  = 0.\end{align*}
\end{enumerate}

\subsection{\texorpdfstring{{Leakage at colluding users in $\mathcal{A} \in \mathbb{A}$, $\mathcal{A} \not\ni l_0$}}{Leakage at colluding users in A in A, A not contains l0}}
 Let $\mathcal{A} \in \mathbb{A}$ such that $\mathcal{A} \not\ni l_0$. Then,
\begin{align*}
& \sum_{l \in \mathcal{A}^{c*}} R_l^U \\
& \stackrel{(a)} \leq \sum_{l \in \mathcal{A}^{c*}} I(U_l;X_l|U_{1:l-1} V_{\mathcal{L}^*}X_{l_0}) + \epsilon L\\
&\stackrel{(b)} \leq \sum_{l \in \mathcal{A}^{c*}} I(U_l;X_l U_{\llbracket 1:l-1 \rrbracket \cap \mathcal{A}}|U_{\llbracket 1:l-1 \rrbracket \cap \mathcal{A}^{c*}} V_{\mathcal{L}^*}X_{l_0})+ \epsilon L\\
& = \smash{\sum_{l \in \mathcal{A}^{c*}}} I(U_l;X_l |U_{\llbracket 1:l-1 \rrbracket \cap \mathcal{A}^{c*}} V_{\mathcal{L}^*}X_{l_0})  \\
&\phantom{----}+ I(U_l; U_{\llbracket 1:l-1 \rrbracket \cap \mathcal{A}}| X_l U_{\llbracket 1:l-1 \rrbracket \cap \mathcal{A}^{c*}} V_{\mathcal{L}^*}X_{l_0})+ \epsilon L\\
& \stackrel{(c)} = \sum_{l \in \mathcal{A}^{c*}} I(U_l;X_l |U_{\llbracket 1:l-1 \rrbracket \cap \mathcal{A}^{c*}} V_{\mathcal{L}^*}X_{l_0}) + \epsilon L\\
& \leq \sum_{l \in \mathcal{A}^{c*}} I(U_l;X_{\mathcal{A}^{c*}} |U_{\llbracket 1:l-1 \rrbracket \cap \mathcal{A}^{c*}} V_{\mathcal{L}^*}X_{l_0}) + \epsilon L\\
& =   I(U_{\mathcal{A}^{c*}};X_{\mathcal{A}^{c*}} |V_{\mathcal{L}^*}X_{l_0})+ \epsilon L, \numberthis \label{eqsumRU}
\end{align*}
where \begin{enumerate}[(a)]
\item  holds by \eqref{eqRU}; \item holds by the chain rule and non-negativity of the mutual information; \item  holds because 
 $I(U_l; U_{\llbracket 1:l-1 \rrbracket \cap \mathcal{A}}| X_l U_{\llbracket 1:l-1 \rrbracket \cap \mathcal{A}^{c*}} V_{\mathcal{L}^*}X_{l_0}) \leq I(U_l V_l; U_{\llbracket 1:l-1 \rrbracket  }  V_{\mathcal{L}^* \backslash \{l\}}X_{l_0}| X_l ) = 0$. 
 \end{enumerate}
 
Then, using the notation  $I_{\tilde{p}}(X_{\mathcal{A}^c}^n ;   U^n_{\mathcal{L}} V^n_{\mathcal{L}}  |  X_{\mathcal{A}}^n  ) $ to reflect that the mutual information considered is with respect to the distribution $\tilde{p}$ induced by the coding scheme, we have 
\begin{align*}
& \frac{1}{n}I_{\tilde{p}}(X_{\mathcal{A}^c}^n ;   U^n_{\mathcal{L}^*} V^n_{\mathcal{L}^*}  |  X_{\mathcal{A}}^n  )  \\
& \stackrel{(a)}= \frac{1}{n}I(X_{\mathcal{A}^c}^n ;   U^n_{\mathcal{L}^*} V^n_{\mathcal{L}^*}  |  X_{\mathcal{A}}^n  ) +o(1) \\
& = I(X_{\mathcal{A}^c} ;   U_{\mathcal{L}^*} V_{\mathcal{L}^*}  |  X_{\mathcal{A}}  ) +o(1) \\
& = I(X_{\mathcal{A}^c} ;   U_{\mathcal{A}^{c*}}  |  X_{\mathcal{A}}  ) + I(X_{\mathcal{A}^c} ;   U_{\mathcal{A}} V_{\mathcal{L}^*}  | U_{\mathcal{A}^{c*}} X_{\mathcal{A}}  ) +  o(1) \\
& \stackrel{(b)} = I(X_{\mathcal{A}^c} ;   U_{\mathcal{A}^{c*}}  |  X_{\mathcal{A}}  )  +  o(1) , \numberthis \label{eqint1}
\end{align*}
where \begin{enumerate}[(a)]
    \item  holds by \cite[Lemma 2.7]{bookCsizar} and $I(X_{\mathcal{A}^c}^n ;   U^n_{\mathcal{L}} V^n_{\mathcal{L}}  |  X_{\mathcal{A}}^n  )$  is considered with respect to the distribution $p$; \item  holds because $I(X_{\mathcal{A}^c} ;   U_{\mathcal{A}} V_{\mathcal{L}^*}  | U_{\mathcal{A}^{c*}} X_{\mathcal{A}}  ) = I(X_{\mathcal{A}^c} ;   U_{\mathcal{A}} V_{\mathcal{A}}  | U_{\mathcal{A}^{c*}} X_{\mathcal{A}}  ) + I(X_{\mathcal{A}^c} ;   V_{\mathcal{A}^{c*}}  | U_{\mathcal{L}^*} V_{\mathcal{A}} X_{\mathcal{A}}  ) \leq I(X_{\mathcal{A}^c} U_{\mathcal{A}^{c*}};   U_{\mathcal{A}} V_{\mathcal{A}}  |  X_{\mathcal{A}}  ) + I(  U_{\mathcal{A}} V_{\mathcal{A}} X_{\mathcal{L}};  V_{\mathcal{A}^{c*}}  | U_{\mathcal{A}^{c*}}   )=0$.
\end{enumerate}
Then, we have
\begin{align*}
&  \frac{1}{n}I_{\tilde{p}}(X_{\mathcal{A}^c}^n ;  U^n_{\mathcal{L}^*} V^n_{\mathcal{L}^*}  | M_{\mathcal{L}^*} X_{\mathcal{A}}^n  )
\\
& \stackrel{(a)}\geq  \frac{1}{n}I_{\tilde{p}}(X_{\mathcal{A}^c}^n ;  U^n_{\mathcal{A}^{c*}} V^n_{\mathcal{L}^*}  | M_{\mathcal{L}^*} X_{\mathcal{A}}^n  )  \\
& \stackrel{(b)}\geq  \frac{1}{n}I_{\tilde{p}}(X_{\mathcal{A}^c}^n ;  U^n_{\mathcal{A}^{c*}} |V^n_{\mathcal{L}^*}   M_{\mathcal{L}^*} X_{\mathcal{A}}^n  )  \\
& =   \frac{1}{n}H_{\tilde{p}}(  U^n_{\mathcal{A}^{c*}} |V^n_{\mathcal{L}^*}   M_{\mathcal{L}^*} X_{\mathcal{A}}^n  ) -  \frac{1}{n}H_{\tilde{p}}(  U^n_{\mathcal{A}^{c*}} |V^n_{\mathcal{L}^*}   M_{\mathcal{L}^*} X_{\mathcal{L}}^n  )  \\
& \stackrel{(c)}=  \frac{1}{n}H_{\tilde{p}}(  U^n_{\mathcal{A}^{c*}} |V^n_{\mathcal{L}^*}   M_{\mathcal{L}^*} X_{\mathcal{A}}^n  )  - o(1)\\
& =  \frac{1}{n}H_{\tilde{p}}(  U^n_{\mathcal{A}^{c*}} |V^n_{\mathcal{L}^*}   X_{\mathcal{A}}^n  ) -\frac{1}{n}I_{\tilde{p}}(  U^n_{\mathcal{A}^{c*}} ; M^U_{\mathcal{L}^*}|V^n_{\mathcal{L}^*}    X_{\mathcal{A}}^n  )   - o(1)\\
& \stackrel{(d)} =  \frac{1}{n}H_{\tilde{p}}(  U^n_{\mathcal{A}^{c*}} |V^n_{\mathcal{L}^*}   X_{\mathcal{A}}^n  ) - \frac{1}{n}I_{\tilde{p}}(  U^n_{\mathcal{A}^{c*}} ; M_{\mathcal{A}^{c*}}^U|V^n_{\mathcal{L}^*}    X_{\mathcal{A}}^n  )   - o(1)\\
& \geq   \frac{1}{n}H_{\tilde{p}}(  U^n_{\mathcal{A}^{c*}} |V^n_{\mathcal{L}^*}   X_{\mathcal{A}}^n  ) - \frac{1}{n}H_{\tilde{p}}(    M_{\mathcal{A}^{c*}}^U  )   - o(1)\\
&  \stackrel{(e)}=  \frac{1}{n}H(  U^n_{\mathcal{A}^{c*}} |V^n_{\mathcal{L}^*}   X_{\mathcal{A}}^n  ) - \frac{1}{n}H_{\tilde{p}}(    M_{\mathcal{A}^{c*}}^U  )   - o(1)\\
&  \stackrel{(f)}\geq H(  U_{\mathcal{A}^{c*}} |V_{\mathcal{L}^*}   X_{\mathcal{A}}  )- I(U_{\mathcal{A}^{c*}};X_{\mathcal{A}^{c*}} |V_{\mathcal{L}^*}X_{l_0})\\
& \phantom{--}- \epsilon L-  o(1)  \\
& \stackrel{(g)}=  H(  U_{\mathcal{A}^{c*}} |V_{\mathcal{A}^{c*}}   X_{\mathcal{A}}  )- I(U_{\mathcal{A}^{c*}};X_{\mathcal{A}^{c*}} |V_{\mathcal{L}^*}X_{l_0})\\
& \phantom{--}- \epsilon L-  o(1)  \\
& \geq  I(  U_{\mathcal{A}^{c*}};  X_{\mathcal{A}^{c*}} |V_{\mathcal{A}^{c*}}   X_{\mathcal{A}}  )- I(U_{\mathcal{A}^{c*}};X_{\mathcal{A}^{c*}} |V_{\mathcal{L}^*}X_{l_0})\\
& \phantom{--}- \epsilon L-  o(1)  \\
& \geq  I(  U_{\mathcal{A}^{c*}};  X_{\mathcal{A}^{c*}} |V_{\mathcal{A}^{c*}}   X_{\mathcal{A}}  )- I(U_{\mathcal{A}^{c*}};X_{\mathcal{A}^{c*}} V_{\mathcal{A}}|V_{\mathcal{A}^{c*}}X_{l_0})\\
& \phantom{--}- \epsilon L-  o(1)  \\
& \stackrel{(h)} =  I(  U_{\mathcal{A}^{c*}};  X_{\mathcal{A}^{c*}} |V_{\mathcal{A}^{c*}}   X_{\mathcal{A}}  )- I(U_{\mathcal{A}^{c*}};X_{\mathcal{A}^{c*}}|V_{\mathcal{A}^{c*}}X_{l_0})\\
& \phantom{--}- \epsilon L-  o(1)  \\
&  =  H(  U_{\mathcal{A}^{c*}}  |V_{\mathcal{A}^{c*}}   X_{\mathcal{A}}  ) -  H(  U_{\mathcal{A}^{c*}} |  X_{\mathcal{A}^{c*}} V_{\mathcal{A}^{c*}}   X_{\mathcal{A}}  )\\
& \phantom{--}- H(U_{\mathcal{A}^{c*}}|V_{\mathcal{A}^{c*}}X_{l_0}) + H(U_{\mathcal{A}^{c*}}|X_{\mathcal{A}^{c*}}V_{\mathcal{A}^{c*}}X_{l_0})\\
& \phantom{--}- \epsilon L-  o(1)  \\
& \stackrel{(i)} =  H(  U_{\mathcal{A}^{c*}}  |V_{\mathcal{A}^{c*}}   X_{\mathcal{A}}  )  - H(U_{\mathcal{A}^{c*}}|V_{\mathcal{A}^{c*}}X_{l_0})    - \epsilon L- o(1)  \\
&  = - I(  U_{\mathcal{A}^{c*}};     X_{\mathcal{A}} |V_{\mathcal{A}^{c*}}   )  + I(U_{\mathcal{A}^{c*}}; X_{l_0}|V_{\mathcal{A}^{c*}})    - \epsilon L- o(1), \numberthis \label{eqint2}
\end{align*}
where 
\begin{enumerate}[(a)]
    \item and (b) hold by the chain rule and non-negativity of the mutual information; \addtocounter{enumi}{1} \item holds by Fano's inequality; \item holds by the chain rule and because 
    \begin{align*} 
&  \frac{1}{n}I_{\tilde{p}}(  U^n_{\mathcal{A}^{c*}} ; M^U_{\mathcal{A}}|M^U_{\mathcal{A}^{c*}}V^n_{\mathcal{L}^*}    X_{\mathcal{A}}^n)\\&\leq \frac{1}{n}I_{\tilde{p}}(  U^n_{\mathcal{A}^{c*}} ; U^n_{\mathcal{A}} V^n_{\mathcal{A}}|M^U_{\mathcal{A}^{c*}}V^n_{\mathcal{L}^*}    X_{\mathcal{A}}^n)\\&\leq \frac{1}{n}I_{\tilde{p}}(  U^n_{\mathcal{A}^{c*}} V^n_{\mathcal{A}^{c*}}; U^n_{\mathcal{A}} V^n_{\mathcal{A}}|     X_{\mathcal{A}}^n)\\&= I (  U_{\mathcal{A}^{c*}} V_{\mathcal{A}^{c*}}; U_{\mathcal{A}} V_{\mathcal{A}}|     X_{\mathcal{A}}) + o(1)= o(1),  \end{align*} where the first equality holds by  \cite[Lemma 2.7]{bookCsizar}; \item holds by  \cite[Lemma 2.7]{bookCsizar}; \item  holds by \eqref{eqsumRU}; \item  holds because $$I(U_{\mathcal{A}^{c*}} ; V_{\mathcal{A}}|V_{\mathcal{A}^{c*}}   X_{\mathcal{A}}) \leq I(U_{\mathcal{A}^{c*}} V_{\mathcal{A}^{c*}} ; V_{\mathcal{A}}|  X_{\mathcal{A}}) = 0 ;$$ \item  holds by the chain rule and because \begin{align*}
&I(U_{\mathcal{A}^{c*}}; V_{\mathcal{A}}|V_{\mathcal{A}^{c*}} X_{\mathcal{A}^{c*}} X_{l_0}) \\&\leq I(U_{\mathcal{A}^{c*}}V_{\mathcal{A}^{c*}}; V_{\mathcal{A}}X_{l_0}| X_{\mathcal{A}^{c*}} ) = 0;\end{align*} \item  holds because \begin{align*}
H(  U_{\mathcal{A}^{c*}} |  X_{\mathcal{A}^{c*}} V_{\mathcal{A}^{c*}}   X_{\mathcal{A}}  ) &= H(  U_{\mathcal{A}^{c*}} |  X_{\mathcal{A}^{c*}} V_{\mathcal{A}^{c*}}    )\\& = H(U_{\mathcal{A}^{c*}}|X_{\mathcal{A}^{c*}}V_{\mathcal{A}^{c*}}X_{l_0}).\end{align*}
    \end{enumerate}

Finally, we have
\begin{align*}
&\frac{1}{n}I_{\tilde{p}}(X_{\mathcal{A}^c}^n ; M_{\mathcal{L}^*} |  X_{\mathcal{A}}^n ) \\
& = \frac{1}{n}I_{\tilde{p}}(X_{\mathcal{A}^c}^n ; M_{\mathcal{L}^*} U^n_{\mathcal{L}^*} V^n_{\mathcal{L}^*}  |  X_{\mathcal{A}}^n  ) \\
& \phantom{--} - \frac{1}{n}I_{\tilde{p}}(X_{\mathcal{A}^c}^n ;  U^n_{\mathcal{L}^*} V^n_{\mathcal{L}^*}  | M_{\mathcal{L}^*} X_{\mathcal{A}}^n  )
\\
& = \frac{1}{n}I_{\tilde{p}}(X_{\mathcal{A}^c}^n ;  U^n_{\mathcal{L}^*} V^n_{\mathcal{L}^*}  |  X_{\mathcal{A}}^n  )   \\
& \phantom{--} - \frac{1}{n}I_{\tilde{p}}(X_{\mathcal{A}^c}^n ;  U^n_{\mathcal{L}^*} V^n_{\mathcal{L}^*}  | M_{\mathcal{L}^*} X_{\mathcal{A}}^n  ) \\
& \leq I(X_{\mathcal{A}^c} ;   U_{\mathcal{A}^{c*}}  |  X_{\mathcal{A}}  )   \\
& \phantom{-}-[I(U_{\mathcal{A}^{c*}}; X_{l_0}|V_{\mathcal{A}^{c*}})- 
\!I(  U_{\mathcal{A}^{c*}};     X_{\mathcal{A}} |V_{\mathcal{A}^{c*}}   ) ]^+  \!+ \epsilon L+
\!o(1),
\end{align*}
where the inequality holds by \eqref{eqint1} and \eqref{eqint2}.

\subsection{Leakage at colluding users in $\mathcal{A} \in \mathbb{A}$, $\mathcal{A} \ni l_0$} Let $\mathcal{A} \in \mathbb{A}$ such that $\mathcal{A} \ni l_0$. Then, we have
\begin{align*}
&\frac{1}{n}I_{\tilde{p}}(X_{\mathcal{A}^c}^n ; M_{\mathcal{L}^*} | F^n X_{\mathcal{A}}^n ) \\
&= \frac{1}{n}I_{\tilde{p}}(X_{\mathcal{A}^c}^n ; M_{\mathcal{L}^*}  F^n |X_{\mathcal{A}}^n )  - \frac{1}{n}I_{\tilde{p}}(X_{\mathcal{A}^c}^n ; F^n| X_{\mathcal{A}}^n )\\
&= \frac{1}{n}I_{\tilde{p}}(X_{\mathcal{A}^c}^n ; M_{\mathcal{L}^*}  |X_{\mathcal{A}}^n ) +\frac{1}{n}I_{\tilde{p}}(X_{\mathcal{A}^c}^n ;   F^n |M_{\mathcal{L}^*} X_{\mathcal{A}}^n ) \\
&\phantom{--}- \frac{1}{n}I_{\tilde{p}}(X_{\mathcal{A}^c}^n ; F^n| X_{\mathcal{A}}^n )\\
&\stackrel{(a)}= \frac{1}{n}I_{\tilde{p}}(X_{\mathcal{A}^c}^n ; M_{\mathcal{L}^*}  |X_{\mathcal{A}}^n )  +\frac{1}{n}I_{\tilde{p}}(X_{\mathcal{A}^c}^n ;   F^n |M_{\mathcal{L}^*} X_{\mathcal{A}}^n ) \\
&\phantom{--}-  I (X_{\mathcal{A}^c}  ; F | X_{\mathcal{A}}  ) + o(1)\\
&\stackrel{(b)}= \frac{1}{n}I_{\tilde{p}}(X_{\mathcal{A}^c}^n ; M_{\mathcal{L}^*}  |X_{\mathcal{A}}^n )   -  I (X_{\mathcal{A}^c}  ; F | X_{\mathcal{A}}  ) + o(1)\\
& \leq  \frac{1}{n}I_{\tilde{p}}(X_{\mathcal{A}^c}^n ;   U^n_{\mathcal{L}^*} V^n_{\mathcal{L}^*}  |  X_{\mathcal{A}}^n  ) -  I (X_{\mathcal{A}^c}  ; F | X_{\mathcal{A}}  ) + o(1)\\
& \stackrel{(c)}= I (X_{\mathcal{A}^c} ;  U_{\mathcal{A}^c}   |  X_{\mathcal{A}}  ) -  I (X_{\mathcal{A}^c}  ; F | X_{\mathcal{A}}  ) + o(1),
\end{align*}
where \begin{enumerate}[(a)]
    \item holds by \cite[Lemma 2.7]{bookCsizar}; \item holds by  Fano's inequality because $l_0 \in \mathcal{A}$; \item  holds similar to \eqref{eqint1}.
    \end{enumerate}

\section{Proof of Theorem \ref{th4}} \label{Secth4}
  The achievability of Theorem \ref{th4} is similar to the achievability scheme in Theorem \ref{thach2}, we thus focus on the converse of Theorem \ref{th4}.
  
For $l \in \mathcal{L}^* \triangleq \mathcal{L} \backslash \{ l_0\}$, we have
\begin{align*}
 R_{l}  
 & \stackrel{(a)}\geq \frac{1}{n}   I(M_{l} ;X_{l}^n )\numberthis \label{eql1cd}\\
& \stackrel{(b)}= \frac{1}{n}  \sum_{i=1}^n I(M_{l} X_{l}^{i-1} ;X_{l,i} )\\
& \stackrel{(c)} =   \frac{1}{n}  \sum_{i=1}^n I( U_{l,i} ;X_{l,i})\\
& \stackrel{(d)} =   \frac{1}{n}  \sum_{i=1}^n I( U_{l,i} ;X_{l,i} | V_i)\\
& \stackrel{(e)} =  I( U_{l};X_{l}|Q) \numberthis \label{eql2cd} ,
\end{align*}
where \begin{enumerate}[(a)]
    \item and (b) hold by \eqref{eql1b} and \eqref{eqrev2}, respectively, with $\mathcal{S} \leftarrow \{ l\}$, $X_{l_0} \leftarrow \emptyset$; \addtocounter{enumi}{1}\item holds with $U_{l,i} \triangleq (M_l, X_{l}^{i-1})$; \item  holds with $V_i \triangleq X_{l_0,\llbracket 1 ,n \rrbracket \backslash \{ i \}}$; \item holds with $U_l \triangleq U_{l,T}$, $Q \triangleq (V_{ T},T)$, and  $T$ uniformly distributed over $\llbracket 1 , n \rrbracket$ and independent of all other random variables, and $X_l \triangleq X_{l,T}$.
    \end{enumerate}Note that  
\begin{align*}
p_{QX_{\mathcal{L}}U_{\mathcal{L}^*}} 
&= p_Q p_{X_{\mathcal{L}}} \prod_{l\in\mathcal{L}^*} p_{U_l|X_lQ} \\
&= p_Q p_{X_{l_0}} \prod_{l\in\mathcal{L}^*} p_{X_{l}} p_{U_l|X_lQ} \\
&= p_Q p_{X_{l_0}} \prod_{l\in\mathcal{L}^*}  p_{U_lX_l|Q}  \numberthis \label{eqindepqd}
\end{align*}
because for $l \in \mathcal{L}^*$, $\llbracket 1, l-1 \rrbracket^*\triangleq \llbracket 1, l-1 \rrbracket \backslash \{ l_0\}$, and $i\in \llbracket 1, n \rrbracket$, we have
\begin{align*}
&I(U_{l,i}; X_{\mathcal{L} \backslash \{ l\},i}U_{\llbracket 1, l-1 \rrbracket^*,i} | X_{l,i} V_{i} )\\
& = I(M_l X_{l}^{i-1}; X_{\mathcal{L} \backslash \{ l\},i} M_{\llbracket 1, l-1 \rrbracket^*} X^{i-1}_{\llbracket 1, l-1 \rrbracket^*} | X_{l,i} X_{l_0,\llbracket 1 ,n \rrbracket \backslash \{ i \}})\\
& \leq I( X_{l}^{n}; X_{\mathcal{L} \backslash \{ l\},i}  X^{n}_{\llbracket 1, l-1 \rrbracket^*} | X_{l,i}X_{l_0,\llbracket 1 ,n \rrbracket \backslash \{ i \}} )\\
& = I( X_{l}^{n}; X_{\mathcal{L} \backslash \{ l\},i}  X^{n}_{\llbracket 1, l-1 \rrbracket^*} | X_{l,i}  )\\
& \leq I( X_{l}^{n}; X_{\mathcal{L} \backslash \{ l\},i}  X^{n}_{\llbracket 1, l-1 \rrbracket^*} )\\
& = 0,
\end{align*}
which implies $0 = I(U_{l,T}; X_{\mathcal{L} \backslash \{ l\},T}U_{\llbracket 1, l-1 \rrbracket^*,T} | X_{l,T} V_T T) =  I(U_{l}; X_{\mathcal{L}\backslash \{ l\}}U_{\llbracket 1, l-1 \rrbracket^*} | X_{l} Q)$.

Next, we have
\begin{align*}
o(n)
&   = H(F^n | \hat{F}^n)\\
&   \stackrel{(a)} \geq H(F^n | M_{\mathcal{L}^*}X^n_{l_0})\\
& = \sum_{i=1}^n H(F_i | F^{i-1} M_{\mathcal{L}^*}X^n_{l_0})\\
& \stackrel{(b)} \geq  \sum_{i=1}^n H(F_i | X_{\mathcal{L}}^{i-1} M_{\mathcal{L}^*}X^n_{l_0})\\
& \stackrel{(c)} =  \sum_{i=1}^n H(F_i | U_{\mathcal{L}^*,i} V_i X_{l_0,i})\\
& = n  H(F | U_{\mathcal{L}^*} X_{l_0} Q) \numberthis \label{eqUbd},
\end{align*}
where \begin{enumerate}[(a)]
    \item and (b) hold by the data processing inequality; \addtocounter{enumi}{1}\item  holds by the definition of $V_i$ and $U_{l,i}$, $l\in\mathcal{L}^*$, $i \in \llbracket 1,n \rrbracket$, and the notation $U_{\mathcal{L}^*,i} \triangleq (U_{l,i})_{l\in\mathcal{L}^*}$.  
    \end{enumerate}
Next, for $\mathcal{A} \in \mathbb{A}$ such that $ \mathcal{A}  \not\ni l_0$,   we have
\begin{align*}
\Delta_{\mathcal{A}}
& \geq \frac{1}{n}I(X_{\mathcal{A}^c}^n ; M_{\mathcal{L}^*} |  X_{\mathcal{A}}^n) \\
& =  \frac{1}{n}I(X_{\mathcal{A}^{c}}^n ; M_{\mathcal{A}^{c*}} |  X_{\mathcal{A}}^n)  \\
& =  \frac{1}{n}I(X_{\mathcal{A}^{c*}}^n ; M_{\mathcal{A}^{c*}} |  X_{\mathcal{A}}^n)  \\
&  \stackrel{(a)} = \frac{1}{n}I(X_{\mathcal{A}^{c*}}^n ; M_{\mathcal{A}^{c*}}   )\\
& \stackrel{(b)} = I(X_{\mathcal{A}^{c*}} ; U_{\mathcal{A}^{c*}}  |Q ) , \numberthis \label{eqconindep2d}
\end{align*}
where \begin{enumerate}[(a)]
     \item holds by independence between $(X_{\mathcal{A}^{c*}}^n , M_{\mathcal{A}^{c*}})$ and $X_{\mathcal{A}}^n$; \item holds  as in the steps between \eqref{eql1cd} and \eqref{eql2cd} since $I(X_{\mathcal{A}^{c*}}^n ; M_{\mathcal{A}^{c*}}   ) = \sum_{l \in \mathcal{A}^{c*}} I(X_{l}^n ; M_{l}   ) = n\sum_{l \in \mathcal{A}^{c*}} I( U_{l};X_{l}|Q) = nI(X_{\mathcal{A}^{c*}} ; U_{\mathcal{A}^{c*}}  |Q ) $ by \eqref{eqindepqd}.
\end{enumerate}

Next, for $\mathcal{A} \in \mathbb{A}$ such that $ \mathcal{A} \ni l_0$, we have
\begin{align*}
\Delta_{\mathcal{A}}
& \geq \frac{1}{n}I(X_{\mathcal{A}^c}^n ; M_{\mathcal{L}^*} |F^n  X_{\mathcal{A}}^n) \\
& = \frac{1}{n}I(X_{\mathcal{A}^c}^n ; M_{\mathcal{L}^*} F^n|  X_{\mathcal{A}}^n)-I(X_{\mathcal{A}^c} ;  F|  X_{\mathcal{A}}) \\
& = \frac{1}{n}I(X_{\mathcal{A}^c}^n ; M_{\mathcal{L}^*} |  X_{\mathcal{A}}^n)+\frac{1}{n}I(X_{\mathcal{A}^c}^n ;  F^n| M_{\mathcal{L}^*}  X_{\mathcal{A}}^n)\\
&\phantom{--}-I(X_{\mathcal{A}^c} ;  F|  X_{\mathcal{A}}) \\
&\stackrel{(a)} =  \frac{1}{n}I(X_{\mathcal{A}^c}^n ; M_{\mathcal{A}^{c}} |  X_{\mathcal{A}}^n) -I(X_{\mathcal{A}^c} ;  F|  X_{\mathcal{A}}) + o(1) \\
& \stackrel{(b)} = I(X_{\mathcal{A}^c} ; U_{\mathcal{A}^c}  |Q ) -I(X_{\mathcal{A}^c} ;  F|  X_{\mathcal{A}}) + o(1) \\
& \stackrel{(c)}= I(X_{\mathcal{A}^c} ; U_{\mathcal{A}^c}  |Q X_{\mathcal{A}}) -I(X_{\mathcal{A}^c} ;  F|  X_{\mathcal{A}} Q) + o(1) \\
& = I(X_{\mathcal{A}^c} ; U_{\mathcal{A}^c} F |Q X_{\mathcal{A}}) - I(X_{\mathcal{A}^c} ;  F |U_{\mathcal{A}^c} Q X_{\mathcal{A}}) \\
&\phantom{--}-I(X_{\mathcal{A}^c} ;  F|  X_{\mathcal{A}} Q) + o(1) \\
& = I(X_{\mathcal{A}^c} ; U_{\mathcal{A}^c}| F Q X_{\mathcal{A}}) - H( F |U_{\mathcal{A}^c} Q X_{\mathcal{A}})  + o(1) \\
& \stackrel{(d)} = I(X_{\mathcal{A}^c} ; U_{\mathcal{A}^c}| F Q X_{\mathcal{A}}) - H( F |U_{\mathcal{L}^*} Q X_{\mathcal{A}})  + o(1) \\
& \stackrel{(e)}= I(X_{\mathcal{A}^c} ; U_{\mathcal{A}^c}| F Q X_{\mathcal{A}}) + o(1) ,   
\end{align*}
where \begin{enumerate}[(a)]
    \item holds by Fano's inequality because $l_0 \in \mathcal{A}$; \item  holds similar to \eqref{eqconindep2d}; \item holds by independence between $X_{\mathcal{A}}$ and $(X_{\mathcal{A}^c} , U_{\mathcal{A}^c}  ,Q)$; \item holds because with $\mathcal{A}^* \triangleq \mathcal{A} \backslash \{ l_0\}$, we have $I(F; U_{\mathcal{A}^*}|U_{\mathcal{A}^c} Q X_{\mathcal{A}})\leq I(X_{\mathcal{L}}; U_{\mathcal{A}^*}|U_{\mathcal{A}^c} Q X_{\mathcal{A}}) = I(X_{\mathcal{A}^c}; U_{\mathcal{A}^*}|U_{\mathcal{A}^c} Q X_{\mathcal{A}}) \leq I(X_{\mathcal{A}^c}U_{\mathcal{A}^c}; X_{\mathcal{A}} U_{\mathcal{A}^*}| Q )= 0$; \item holds by \eqref{eqUbd} because conditioning reduces entropy and $l_0 \in \mathcal{A}$. \end{enumerate}

Finally, the cardinality bounds on $\mathcal{U}_l$, $l\in\mathcal{L}$, and $Q$ follow from the Fenchel-Eggleston-Carath\'{e}odory theorem as in \cite[Appendix A]{courtade2013multiterminal}.

\section{Proof of Theorem \ref{th1}} \label{secth1}
 For $l\in\mathcal{L}$ and $j \in \llbracket 1, n \rrbracket$, we write $X_{l}^j \triangleq (X_{l,i})_{i\in \llbracket 1, j \rrbracket}$. And, for $\mathcal{S} \subseteq \mathcal{L}$, we define $X_{\mathcal{S},i} \triangleq (X_{l,i})_{l\in\mathcal{S}}$ and $X_{\mathcal{S}}^{i-1} \triangleq (X_{l}^{i-1})_{l\in\mathcal{S}}$. 

Then, we have
\begin{align*}
o(n)
& \stackrel{(a)} = H(F^n | \hat{F}^n)\\
&\stackrel{(b)} \geq H(F^n | M_{\mathcal{L}})\\
& = \sum_{i=1}^n H(F_i | F^{i-1} M_{\mathcal{L}})\\
& \stackrel{(c)} \geq  \sum_{i=1}^n H(F_i | X_{\mathcal{L}}^{i-1} M_{\mathcal{L}}) \numberthis \label{equdist}\\
& \stackrel{(d)} =  \sum_{i=1}^n H(F_i | U_{\mathcal{L},i})\\
& \stackrel{(e)} = n  H(F | U_{\mathcal{L}} )\numberthis \label{eqU},
\end{align*}
where \begin{enumerate}[(a)]
    \item holds by Fano's inequality and \eqref{eqrec}; \item  and (c) hold by the data processing inequality; \addtocounter{enumi}{1} \item  holds with $U_{l,i} \triangleq (M_l, X_{\mathcal{L}}^{i-1})$, $i\in \llbracket 
    1,n\rrbracket$, $l\in \mathcal{L}$, and the notation $U_{\mathcal{S},i} \triangleq (U_{l,i})_{l\in\mathcal{S}}$, $\mathcal{S} \subseteq \mathcal{L}$; \item holds with  $U_l \triangleq (Q,U_{l,Q})$ where $Q$ is uniformly distributed over $\llbracket 1 , n \rrbracket$ and independent of all other random variables.
\end{enumerate}
 
Then, for $\mathcal{S} \subseteq \mathcal{L}$, we have
\begin{align*}
 R_{\mathcal{S}}  
&\stackrel{(a)} \geq \frac{1}{n} \sum_{l\in \mathcal{S}} H(M_l) \\
&\geq  \frac{1}{n}  H(M_{\mathcal{S}}) \\
& = \frac{1}{n}  I(M_{\mathcal{S}};X_{\mathcal{S}}^n) \numberthis \label{eql1}\\
& \stackrel{(b)} = \frac{1}{n}  \sum_{i=1}^n I(M_{\mathcal{S}};X_{\mathcal{S},i} | X_{\mathcal{S}}
^{i-1} )\\
& \stackrel{(c)}= \frac{1}{n}  \sum_{i=1}^n I(M_{\mathcal{S}} X_{\mathcal{S}}^{i-1};X_{\mathcal{S},i} ) \numberthis \label{eqratel}\\
&  = \frac{1}{n}  \sum_{i=1}^n I(M_{\mathcal{S}} X_{\mathcal{L}}^{i-1};X_{\mathcal{S},i} ) - I(X_{\mathcal{S}^c}^{i-1};X_{\mathcal{S},i} | M_{\mathcal{S}} X_{\mathcal{S} }^{i-1}) \\
&  \stackrel{(d)} = \frac{1}{n}  \sum_{i=1}^n I(M_{\mathcal{S}} X_{\mathcal{L}}^{i-1};X_{\mathcal{S},i} ) \\
& \stackrel{(e)} =   \frac{1}{n}  \sum_{i=1}^n I( U_{\mathcal{S},i} ;X_{\mathcal{S},i} )\\
&  \stackrel{(f)} =    I( U_{\mathcal{S}};X_{\mathcal{S}}) \numberthis \label{eql2}\\
& =  I( U_{\mathcal{S}} ; X_{\mathcal{S}} U_{\mathcal{S}^c}) - I(U_{\mathcal{S}};U_{\mathcal{S}^c}|X_{\mathcal{S}})  \\
& \geq  I( U_{\mathcal{S}} ; X_{\mathcal{S}}| U_{\mathcal{S}^c}) - I(U_{\mathcal{S}};U_{\mathcal{S}^c}|X_{\mathcal{S}})  ,
\end{align*}
where \begin{enumerate}[(a)]
    \item holds by Definition \ref{def1}; \item holds by the chain rule; \item holds by independence between  $X_{\mathcal{S}}^{i-1} $ and $X_{\mathcal{S},i}$; \item holds because \begin{align*}
I(X_{\mathcal{S}^c}^{i-1};X_{\mathcal{S},i} |M_{\mathcal{S}} X_{\mathcal{S} }^{i-1}) &\leq I(X_{\mathcal{S}^c}^{i-1};X_{\mathcal{S},i} M_{\mathcal{S}}
| X_{\mathcal{S} }^{i-1}) \\& \leq I(X_{\mathcal{S}^c}^{i-1};X_{\mathcal{S}}^n 
| X_{\mathcal{S} }^{i-1})=0;\end{align*} \item holds by definition of $U_{l,i}$, $l\in\mathcal{L}$, $i\in \llbracket 1 ,n \rrbracket$; \item holds by definition of $U_l$, and $X_{\mathcal S} \triangleq X_{ \mathcal S,Q}$.
\end{enumerate}

Next, by \eqref{eqle} we have
\begin{align*}
\Delta
&  \geq \frac{1}{n}I(X_{\mathcal{L}}^n ; M_{\mathcal{L}} |  F^n) \\
& =  \frac{1}{n}I(X_{\mathcal{L}}^n ; M_{\mathcal{L}}   F^n) - \frac{1}{n}I(X_{\mathcal{L}}^n ;   F^n)\\
& =  \frac{1}{n}I(X_{\mathcal{L}}^n ; M_{\mathcal{L}}   ) + \frac{1}{n}I(X_{\mathcal{L}}^n ;F^n| M_{\mathcal{L}}  ) - I(X_{\mathcal{L}} ;   F)\\
& \stackrel{(a)}=  \frac{1}{n}I(X_{\mathcal{L}}^n ; M_{\mathcal{L}}   )  - I(X_{\mathcal{L}} ;   F) + o(1) \numberthis \label{eqdeltac}\\
& \stackrel{(b)}=    I( U_{\mathcal{L}} ; X_{\mathcal{L}}   )  - I(X_{\mathcal{L}} ;   F) + o(1)\\
& =    I( U_{\mathcal{L}} ; X_{\mathcal{L}}  F )   - I(X_{\mathcal{L}} ;   F) + o(1) \\
&  =    I( U_{\mathcal{L}} ; X_{\mathcal{L}} | F ) + I( U_{\mathcal{L}} ;  F )   - I(X_{\mathcal{L}} ;   F) + o(1)\\
& \stackrel{(c)} =    I( U_{\mathcal{L}} ; X_{\mathcal{L}} | F )   + o(1),
\end{align*}
where \begin{enumerate}[(a)]
    \item holds by Fano's inequality and \eqref{eqrec}; \item holds by the steps between \eqref{eql1} and \eqref{eql2} by choosing $\mathcal{S} = \mathcal{L}$; \item holds because $I( U_{\mathcal{L}} ;  F )   - I(X_{\mathcal{L}} ;   F) = H(F|X_{\mathcal{L}})-H(F|U_{\mathcal{L}})= -H(F|U_{\mathcal{L}}) = o(1)$ by~\eqref{eqU}.
    \end{enumerate}

Finally, by \eqref{eql2a}  we have for $\mathcal{A} \in \mathbb{A}$
\begin{align*}
\Delta_{\mathcal{A}}
& \geq \frac{1}{n}I(X_{\mathcal{A}^c}^n ; M_{\mathcal{L}} |  X_{\mathcal{A}}^n) \\
& = \frac{1}{n}I(X_{\mathcal{A}^c}^n ; M_{\mathcal{A}^{c}} |  X_{\mathcal{A}}^n) + \frac{1}{n}I(X_{\mathcal{A}^c}^n ; M_{\mathcal{A}} |  M_{\mathcal{A}^{c}} X_{\mathcal{A}}^n) \\
&= \frac{1}{n}I(X_{\mathcal{A}^c}^n ; M_{\mathcal{A}^{c}} |  X_{\mathcal{A}}^n) \numberthis \label{eqdeltaa3} \\
&  \stackrel{(a)} = \frac{1}{n}I(X_{\mathcal{A}^c}^n ; M_{\mathcal{A}^{c}}   ) - \frac{1}{n}I(    X_{\mathcal{A}}^n ;M_{\mathcal{A}^{c}})\\
& \stackrel{(b)} = I(X_{\mathcal{A}^c} ; U_{\mathcal{A}^c}   ) - \frac{1}{n}I(    X_{\mathcal{A}}^n ;M_{\mathcal{A}^{c}})\\
& = I(X_{\mathcal{A}^c} ; U_{\mathcal{A}^c}   ) - \frac{1}{n} \sum_{i=1}^n I( X_{\mathcal{A},i} ;M_{\mathcal{A}^c} |X_{\mathcal{A}}^{i-1}) \\
&\stackrel{(c)} = I(X_{\mathcal{A}^c} ; U_{\mathcal{A}^c}   ) - \frac{1}{n} \sum_{i=1}^n I( X_{\mathcal{A},i} ;M_{\mathcal{A}^c} X_{\mathcal{A}}^{i-1})\\
& \geq  I(X_{\mathcal{A}^c} ; U_{\mathcal{A}^c}   ) - \frac{1}{n} \sum_{i=1}^n I( X_{\mathcal{A},i} ;M_{\mathcal{A}^c}X_{\mathcal{L}}^{i-1} )\\
& = I(X_{\mathcal{A}^c} ; U_{\mathcal{A}^c}   ) - \frac{1}{n} \sum_{i=1}^n I( X_{\mathcal{A},i} ;U_{\mathcal{A}^c,i} )\\
& = I(X_{\mathcal{A}^c} ; U_{\mathcal{A}^c}   ) - I(X_{\mathcal{A}} ; U_{\mathcal{A}^c}   )\\
& \stackrel{(d)} = I(X_{\mathcal{A}^c} ; U_{\mathcal{A}^c} | X_{\mathcal{A}}   ), 
\end{align*}
where \begin{enumerate}[(a)]
    \item holds because $ M_{\mathcal{A}^c} - X_{\mathcal{A}^c}^n  -  X_{\mathcal{A}}^n$; \item holds by the steps between \eqref{eql1} and \eqref{eql2} by choosing $\mathcal{S} = \mathcal{A}^c$; \item holds by independence between  $X_{\mathcal{A}}^{i-1}$ and $X_{\mathcal{A},i}$; \item holds because $U_{\mathcal{A}^c} -X_{\mathcal{A}^c}-  X_{\mathcal{A}}$ forms a Markov chain since for any $\mathcal{S} \subseteq \mathcal{L}$, $U_{\mathcal{S}} - X_{\mathcal{S}} - X_{\mathcal{L}}$ forms a Markov chain because for any $i \in \llbracket 1, n \rrbracket$, we have
$
I(U_{\mathcal{S},i}; X_{\mathcal{L},i}|X_{\mathcal{S},i})
 = I(M_{\mathcal{S}} X_{\mathcal{L}}^{i-1} ; X_{\mathcal{L},i}|X_{\mathcal{S},i}) 
 \leq I(X_{\mathcal{S}}^n X_{\mathcal{L}}^{i-1} ; X_{\mathcal{L},i}|X_{\mathcal{S},i}) 
 = 0,$ which implies $0 =I(U_{\mathcal{S},Q}; X_{\mathcal{L},Q}|X_{\mathcal{S},Q}Q) = I(U_{\mathcal{S},Q}Q; X_{\mathcal{L},Q}|X_{\mathcal{S},Q}) = I(U_{\mathcal{S}}; X_{\mathcal{L}}|X_{\mathcal{S}})$.
\end{enumerate}

\section{Proof of Theorem \ref{th2}} \label{secth2}

The proof of Theorem \ref{th2} is simpler and closely follows that of Theorem \ref{thach2}. To avoid redundancy, we present only a sketch of the proof for Theorem \ref{th2}.

For $p_{U_{\mathcal{L}}  X_{\mathcal{L}}} = p_{X_{\mathcal{L}}} \prod_{l\in\mathcal{L}}  p_{U_l|X_l} $, one can first prove  that the achievability of $\mathcal{R} (\mathbb{A},{p_{U_{\mathcal{L}}  X_{\mathcal{L}}}})$  can be reduced to the achievability of the rate-tuple 
\begin{align}
\mathbf{R}^{\star}(\mathbb{A},{p_{U_{\mathcal{L}}  X_{\mathcal{L}}}}) \triangleq ((I(U_{l};X_{l} | U_{1:l-1}) )_{l \in \mathcal{L}},\Delta, (\Delta_{\mathcal{A}})_{\mathcal{A} \in \mathbb{A}}), \label{eqredux}
\end{align}
 where, for $l \in \mathcal{L}$, we use the notation $U_{1:l-1}  \triangleq U_{\llbracket 1, l-1 \rrbracket } $ for convenience. The proof is similar to the proof detailed in Section \ref{secreduction2} and is thus omitted. Then, we provide a coding scheme to achieve the rate-tuple $\mathbf{R}^{\star}(\mathbb{A},{p_{U_{\mathcal{L}}  X_{\mathcal{L}}}})$. The analysis of this coding scheme is similar to the proof  in Appendix \ref{appanalysis} and is thus omitted.

 Fix $p_{U_{\mathcal{L}}  X_{\mathcal{L}}} = p_{X_{\mathcal{L}}} \prod_{l\in\mathcal{L}}  p_{U_l|X_l} $ such that $H(F|U_{\mathcal{L}})=0$, and consider $(U^n_{\mathcal{L}}  ,X^n_{\mathcal{L}})$ independently and identically distributed according to $p_{ U_{\mathcal{L}}  X_{\mathcal{L}}}$. For $l\in\mathcal{L}$, to each $u^n_l$ assign uniformly and independently two random bin indices $m_l \in \llbracket 1,2^{nR_l} \rrbracket$ and $j_l\in \llbracket 1,2^{n\tilde{R}_l} \rrbracket$. Let $p_{U^n_{\mathcal{L}}X^n_{\mathcal{L}}M_{\mathcal{L}}J_{\mathcal{L}}} \triangleq p_{ U^n_{\mathcal{L}}  X^n_{\mathcal{L}}} \prod_{l\in\mathcal{L}}p_{M_lJ_l|U^n_l}$ be the distribution induced by this random binning. Note that we have
$$p_{U^n_{\mathcal{L}}X^n_{\mathcal{L}}M_{\mathcal{L}}J_{\mathcal{L}}}
 = p_{ X^n_{\mathcal{L}}} \prod_{l\in\mathcal{L}}p_{J_l|X^n_l}  p_{U^n_l|J_l X^n_l}  p_{M_l|U^n_l} .
$$ Hence, by Lemma \ref{lembinindep} with $\tilde{R}_l = H(U_l|X_l) - \epsilon= H(U_l|X_lU_{1:l-1}) - \epsilon$, $l\in\mathcal{L}$, $\epsilon >0$, we have $\lim_{n\to\infty} \mathbb{E}[\mathbb{V}(p_{U^n_{\mathcal{L}}X^n_{\mathcal{L}}M_{\mathcal{L}}J_{\mathcal{L}}},\tilde{p}_{U^n_{\mathcal{L}}X^n_{\mathcal{L}}M_{\mathcal{L}}J_{\mathcal{L}}} )] = 0$, where the average is over the random binning choice and 
$$\tilde{p}_{U^n_{\mathcal{L}}X^n_{\mathcal{L}}M_{\mathcal{L}}J_{\mathcal{L}}} \triangleq p_{ X^n_{\mathcal{L}}} \prod_{l\in\mathcal{L}}  p^{\text{unif}}_{J_l}p_{U^n_l|J_lX^n_l}p_{M_l|U^n_l},$$ with $p^{\text{unif}}_{J_l}$ the uniform distribution over $\llbracket 1,2^{n\tilde{R}_l} \rrbracket$, $l\in\mathcal{L}$.

Assume now that all the users have access to the common randomness $(J_l)_{l\in\mathcal{L}}$ uniformly distributed over $\bigtimes_{l\in\mathcal{L}} \llbracket 1, 2^{n\tilde{R}_l} \rrbracket$. User $l\in\mathcal{L}$ generates $U_l^n$ according to $p_{U^n_l|J_lX^n_l}$, then transmits $M_l$ generated according to $p_{M_l|U^n_l}$. By Lemma \ref{lemSW}, the fusion center can reconstruct with vanishing probability of error  $U^n_{\mathcal{L}}$ from $(M_l,J_l)_{l\in \mathcal{L}}$ by choosing for $l\in\mathcal{L}$, $\tilde{R}_l + R_l = H(U_l|U_{1:l-1})+\epsilon$, i.e., $ R_l = I(U_l;X_l|U_{1:l-1})+2\epsilon$. Then, since $H(F|U_{\mathcal{L}})=0$, there exists a deterministic function $\tilde{F}$ such that $\tilde{F}(u_{\mathcal{L}})=F(x_{\mathcal{L}})$, for all $(u_{\mathcal{L}},x_{\mathcal{L}}) $ such that $p(u_{\mathcal{L}},x_{\mathcal{L}})>0$. Finally,  the fusion center computes $\widehat{F}^n \triangleq \tilde{F}(\hat{U}_{\mathcal{L}}^n)$.

\section{Proof of Theorem \ref{th3}} \label{Secth3}
 The achievability of Theorem \ref{th3} follows from the achievability scheme in Theorem \ref{th2} and the introduction of the variable $Q$, we thus focus on the converse of Theorem \ref{th3}.

For $l \in \mathcal{L}$, we have
\begin{align*}
 R_{l}  
 & \stackrel{(a)}\geq \frac{1}{n}  I(M_{l} ;X_{l}^n )\numberthis \label{eql1c}\\
& \stackrel{(b)}= \frac{1}{n}  \sum_{i=1}^n I(M_{l} X_{l}^{i-1} ;X_{l,i} )\\
& \stackrel{(c)} =   \frac{1}{n}  \sum_{i=1}^n I( U_{l,i} ;X_{l,i})\\
& \stackrel{(d)} =  I( U_{l};X_{l}|Q) \numberthis \label{eql2c} ,
\end{align*}
where \begin{enumerate}[(a)]
    \item and (b) hold by \eqref{eql1} and \eqref{eqratel}, respectively, with the choice $\mathcal{S} = \{ l\}$;\addtocounter{enumi}{1} \item  holds with $U_{l,i} \triangleq (M_l, X_{l}^{i-1})$; \item holds with $U_l \triangleq U_{l,Q}$ with $Q$ uniformly distributed over $\llbracket 1 , n \rrbracket$ and independent of all other random variables, and $X_l \triangleq X_{l,Q}$. 
    \end{enumerate} Note that we have
\begin{align}
p_{QX_{\mathcal{L}}U_{\mathcal{L}}} = p_Q p_{X_{\mathcal{L}}} \prod_{l\in\mathcal{L}} p_{U_l|X_lQ}  = p_Q \prod_{l\in\mathcal{L}} p_{U_lX_l|Q} \label{eqindepq}
\end{align}
because for $l \in \mathcal{L}$ and $i\in \llbracket 1, n \rrbracket$, we have
\begin{align*}
&I(U_{l,i}; X_{\mathcal{L}\backslash \{ l\},i}U_{\llbracket 1, l-1 \rrbracket,i} | X_{l,i} ) \\
& = I(M_l X_{l}^{i-1}; X_{\mathcal{L}\backslash \{ l\},i} M_{\llbracket 1, l-1 \rrbracket} X^{i-1}_{\llbracket 1, l-1 \rrbracket} | X_{l,i} )\\
& \leq I( X_{l}^{n}; X_{\mathcal{L}\backslash \{ l\},i}  X^{n}_{\llbracket 1, l-1 \rrbracket} | X_{l,i} )\\
& \leq I( X_{l}^{n}; X_{\mathcal{L}\backslash \{ l\},i}  X^{n}_{\llbracket 1, l-1 \rrbracket} )\\
& = 0,
\end{align*}
which implies $0 = I(U_{l,Q}; X_{\mathcal{L}\backslash \{ l\},Q}U_{\llbracket 1, l-1 \rrbracket,Q} | X_{l,Q} Q) =  I(U_{l}; X_{\mathcal{L}\backslash \{ l\}}U_{\llbracket 1, l-1 \rrbracket} | X_{l} Q)$.

Next, we have
\begin{align*}
o(n)
& \stackrel{(a)} \geq  \sum_{i=1}^n H(F_i | X_{\mathcal{L}}^{i-1} M_{\mathcal{L}})\\
& \stackrel{(b)} =  \sum_{i=1}^n H(F_i | U_{\mathcal{L},i} )\\
& = n  H(F | U_{\mathcal{L}} Q) \numberthis \label{eqUb3},
\end{align*}
where \begin{enumerate}[(a)]
    \item holds as in \eqref{equdist}; \item holds by the definition of $U_{l,i}$, $l\in\mathcal{L}$, $i \in \llbracket 1,n \rrbracket$, and the notation $U_{\mathcal{L},i} \triangleq (U_{l,i})_{l\in\mathcal{L}}$. 
    \end{enumerate}Next, we have
\begin{align*}
\Delta
& \stackrel{(a)}\geq   \frac{1}{n}I(X_{\mathcal{L}}^n ; M_{\mathcal{L}}   )  - I(X_{\mathcal{L}} ;   F) + o(1) \numberthis \label{eqUmults0}\\
& \stackrel{(b)}=     \sum_{l\in\mathcal{L}} I(U_{l} ; X_{l}|Q)  - I(X_{\mathcal{L}} ;   F) + o(1) \\
& \stackrel{(c)}=    I( U_{\mathcal{L}} ; X_{\mathcal{L}}  |Q )  - I(X_{\mathcal{L}} ;   F) + o(1) \numberthis \label{eqUmults}\\
& =    I( U_{\mathcal{L}} ; X_{\mathcal{L}}  F |Q)   - I(X_{\mathcal{L}} ;   F) + o(1) \\
&  =    I( U_{\mathcal{L}} ; X_{\mathcal{L}} | F Q) + I( U_{\mathcal{L}} ;  F |Q)   - I(X_{\mathcal{L}} ;   F) + o(1)\\
& \stackrel{(d)} =    I( U_{\mathcal{L}} ; X_{\mathcal{L}} | FQ )   + o(1),
\end{align*}
where \begin{enumerate}[(a)]
    \item holds by \eqref{eqdeltac}; \item  holds because $\frac{1}{n} I(X_{\mathcal{L}}^n ; M_{\mathcal{L}}   ) =  \frac{1}{n}\sum_{l\in\mathcal{L}} I(X_{l}^n ; M_{l}) = \sum_{l\in\mathcal{L}} I(U_{l} ; X_{l}|Q)$  by independence between $((X_l^n,M_l))_{l\in \mathcal{L}}$, and  the steps between \eqref{eql1c} and \eqref{eql2c}; \item holds by the conditional independence of $((X_l,U_l))_{l\in \mathcal{L}}$ given $Q$ shown in \eqref{eqindepq}; \item holds because $I( U_{\mathcal{L}} ;  F |Q)   - I(X_{\mathcal{L}} ;   F)= I( U_{\mathcal{L}} ;  F |Q)   - I(X_{\mathcal{L}} ;   F|Q)  = H(F|X_{\mathcal{L}}Q)-H(F|U_{\mathcal{L}}Q)= -H(F|U_{\mathcal{L}}Q) = o(1)$ by~\eqref{eqUb3}.
    \end{enumerate}
Next, for $\mathcal{A} \in \mathbb{A}$, we have
\begin{align*}
\Delta_{\mathcal{A}}
&\stackrel{(a)} \geq  \frac{1}{n}I(X_{\mathcal{A}^c}^n ; M_{\mathcal{A}^{c}} |  X_{\mathcal{A}}^n)  \\
&  \stackrel{(b)} = \frac{1}{n}I(X_{\mathcal{A}^c}^n ; M_{\mathcal{A}^{c}}   )\\
& \stackrel{(c)} = I(X_{\mathcal{A}^c} ; U_{\mathcal{A}^c}  |Q ) , 
\end{align*}
where \begin{enumerate}[(a)]
    \item holds as in \eqref{eqdeltaa3}; \item holds by independence between $(X_{\mathcal{A}^c}^n , M_{\mathcal{A}^{c}})$ and $X_{\mathcal{A}}^n$; \item holds  as in the steps between \eqref{eqUmults0} and \eqref{eqUmults}.
\end{enumerate}

Finally, the cardinality bounds on $\mathcal{U}_l$, $l\in\mathcal{L}$, and $Q$ follows from the Fenchel-Eggleston-Carath\'{e}odory theorem as in \cite[Appendix A]{courtade2013multiterminal}.

\bibliographystyle{IEEEtran}
\bibliography{bib}

\begin{thebibliography}{10}
\providecommand{\url}[1]{#1}
\csname url@samestyle\endcsname
\providecommand{\newblock}{\relax}
\providecommand{\bibinfo}[2]{#2}
\providecommand{\BIBentrySTDinterwordspacing}{\spaceskip=0pt\relax}
\providecommand{\BIBentryALTinterwordstretchfactor}{4}
\providecommand{\BIBentryALTinterwordspacing}{\spaceskip=\fontdimen2\font plus
\BIBentryALTinterwordstretchfactor\fontdimen3\font minus
  \fontdimen4\font\relax}
\providecommand{\BIBforeignlanguage}[2]{{%
\expandafter\ifx\csname l@#1\endcsname\relax
\typeout{** WARNING: IEEEtran.bst: No hyphenation pattern has been}%
\typeout{** loaded for the language `#1'. Using the pattern for}%
\typeout{** the default language instead.}%
\else
\language=\csname l@#1\endcsname
\fi
#2}}
\providecommand{\BIBdecl}{\relax}
\BIBdecl

\bibitem{chou2022function}
R.~A. Chou and J.~Kliewer, ``Function computation without secure links:
  {I}nformation and leakage rates,'' in \emph{IEEE International Symposium on
  Information Theory (ISIT)}, 2022, pp. 1223--1228.

\bibitem{cramer2015secure}
R.~Cramer, I.~Damg{\aa}rd, and J.~Nielsen, \emph{Secure Multiparty
  Computation}.\hskip 1em plus 0.5em minus 0.4em\relax Cambridge University
  Press, 2015.

\bibitem{orlitsky2001coding}
A.~Orlitsky and J.~Roche, ``Coding for computing,'' \emph{IEEE Transactions on
  Information Theory}, vol.~47, no.~3, pp. 903--917, 2001.

\bibitem{sefidgaran2016distributed}
M.~Sefidgaran and A.~Tchamkerten, ``Distributed function computation over a
  rooted directed tree,'' \emph{IEEE Transactions on Information Theory},
  vol.~62, no.~12, pp. 7135--7152, 2016.

\bibitem{ma2011some}
N.~Ma and P.~Ishwar, ``Some results on distributed source coding for
  interactive function computation,'' \emph{IEEE Transactions on Information
  Theory}, vol.~57, no.~9, pp. 6180--6195, 2011.

\bibitem{ma2012interactive}
N.~Ma, P.~Ishwar, and P.~Gupta, ``Interactive source coding for function
  computation in collocated networks,'' \emph{IEEE Transactions on Information
  Theory}, vol.~58, no.~7, pp. 4289--4305, 2012.

\bibitem{ma2013infinite}
N.~Ma and P.~Ishwar, ``The infinite-message limit of two-terminal interactive
  source coding,'' \emph{IEEE Transactions on Information Theory}, vol.~59,
  no.~7, pp. 4071--4094, 2013.

\bibitem{data2016communication}
D.~Data, V.~M. Prabhakaran, and M.~M. Prabhakaran, ``Communication and
  randomness lower bounds for secure computation,'' \emph{IEEE Transactions on
  Information Theory}, vol.~62, no.~7, pp. 3901--3929, 2016.

\bibitem{lee2014two}
E.~J. Lee and E.~Abbe, ``Two {S}hannon-type problems on secure multi-party
  computations,'' in \emph{52nd Annual Allerton Conference on Communication,
  Control, and Computing (Allerton)}, 2014, pp. 1287--1293.

\bibitem{chou2024private}
R.~A. Chou, J.~Kliewer, and A.~Yener, ``Private sum computation: {T}rade-off
  between shared randomness and privacy,'' in \emph{IEEE International
  Symposium on Information Theory}, 2024, pp. 927--932.

\bibitem{morteza2024distributed}
A.~Morteza and R.~A. Chou, ``Distributed matrix multiplication: {D}ownload
  rate, randomness and privacy trade-offs,'' in \emph{60th Annual Allerton
  Conference on Communication, Control, and Computing}, 2024, pp. 1--7.

\bibitem{data2020interactive}
D.~Data, G.~R. Kurri, J.~Ravi, and V.~M. Prabhakaran, ``Interactive secure
  function computation,'' \emph{IEEE Transactions on Information Theory},
  vol.~66, no.~9, pp. 5492--5521, 2020.

\bibitem{tyagi2011function}
H.~Tyagi, P.~Narayan, and P.~Gupta, ``When is a function securely computable?''
  \emph{IEEE Transactions on Information Theory}, vol.~57, no.~10, pp.
  6337--6350, 2011.

\bibitem{tu2019function}
W.~Tu and L.~Lai, ``On function computation with privacy and secrecy
  constraints,'' \emph{IEEE Transactions on Information Theory}, vol.~65,
  no.~10, pp. 6716--6733, 2019.

\bibitem{Yamamoto1983}
H.~Yamamoto, ``A source coding problem for sources with additional outputs to
  keep secret from the receiver or wiretappers (corresp.),'' \emph{IEEE
  Transactions on Information Theory}, vol.~29, pp. 918--923, November 1983.

\bibitem{Yamamoto1997}
------, ``Rate-distortion theory for the shannon cipher system,'' \emph{IEEE
  Transactions on Information Theory}, vol.~43, pp. 827--835, May 1997.

\bibitem{Merhav2008}
N.~Merhav, ``Shannon's secrecy system with informed receivers and its
  application to systematic coding for wiretapped channels,'' \emph{IEEE
  Transactions on Information Theory}, vol.~54, no.~6, pp. 2723--2734, June
  2008.

\bibitem{Schieler2014}
C.~Schieler and P.~Cuff, ``Rate-distortion theory for secrecy systems,''
  \emph{IEEE Transactions on Information Theory}, vol.~60, no.~12, pp.
  7584--7605, December 2014.

\bibitem{Villard2013}
J.~Villard and P.~Piantanida, ``Secure multiterminal source coding with side
  information at the eavesdropper,'' \emph{IEEE Transactions on Information
  Theory}, vol.~59, no.~6, pp. 3668--3692, June 2013.

\bibitem{Naghibi2015}
F.~Naghibi, S.~Salimi, and M.~Skoglund, ``The {CEO} problem with secrecy
  constraints,'' \emph{IEEE Transactions on Information Forensics and
  Security}, vol.~10, no.~6, pp. 1234--1249, June 2015.

\bibitem{zivarifard2024secure}
H.~ZivariFard and R.~A. Chou, ``Secure source coding resilient against
  compromised users via an access structure,'' \emph{IEEE Journal on Selected
  Areas in Information Theory}, 2024.

\bibitem{edmonds1970submodular}
J.~Edmonds, ``Submodular functions, matroids, and certain polyhedra,
  combinatorial structures and their applications, {R. Guy, H. Hanani, N.
  Sauer, and J. Schonheim}, eds,'' \emph{New York}, pp. 69--87, 1970.

\bibitem{yassaee2014achievability}
M.~H. Yassaee, M.~R. Aref, and A.~Gohari, ``Achievability proof via output
  statistics of random binning,'' \emph{IEEE Transactions on Information
  Theory}, vol.~60, no.~11, pp. 6760--6786, 2014.

\bibitem{bennett2002generalized}
C.~H. Bennett, G.~Brassard, C.~Cr{\'e}peau, and U.~M. Maurer, ``Generalized
  privacy amplification,'' \emph{IEEE Transactions on Information theory},
  vol.~41, no.~6, pp. 1915--1923, 2002.

\bibitem{el2011network}
A.~El~Gamal and Y.-H. Kim, \emph{Network Information Theory}.\hskip 1em plus
  0.5em minus 0.4em\relax Cambridge University Press, 2011.

\bibitem{vamoua}
V.~Yachongka and R.~Chou, ``Privacy-constrained lossy function computation in
  distributed databases,'' in \emph{IEEE International Symposium on Information
  Theory}, 2025.

\bibitem{ahlswede1993common}
R.~Ahlswede and I.~Csisz{\'a}r, ``Common randomness in information theory and
  cryptography. {I. S}ecret sharing,'' \emph{IEEE Transactions on Information
  Theory}, vol.~39, no.~4, pp. 1121--1132, 1993.

\bibitem{bookCsizar}
I.~Csisz\'{a}r and J.~K\"{o}rner, \emph{Information Theory: Coding Theorems for
  Discrete Memoryless Systems}.\hskip 1em plus 0.5em minus 0.4em\relax
  Cambridge Univ Pr, 1981.

\bibitem{courtade2013multiterminal}
T.~A. Courtade and T.~Weissman, ``Multiterminal source coding under logarithmic
  loss,'' \emph{IEEE Transactions on Information Theory}, vol.~60, no.~1, pp.
  740--761, 2013.

\end{thebibliography}

\end{document}